\documentclass[11pt, english, american]{article}

\usepackage{amsfonts,amsmath,amssymb,mathrsfs}
\usepackage{dsfont}
\usepackage{graphicx}% Include figure files
\usepackage{color}

\newcommand{\bet}{\beta_1}
\newcommand{\be}{\begin{equation}}
\newcommand{\ee}{\end{equation}}

\newcommand{\llaa}{\left\langle\hspace{-0.14cm}\left\langle}
\newcommand{\rraa}{\right\rangle\hspace{-0.14cm}\right\rangle}
\newcommand{\laa}{\langle\hspace{-0.08cm}\langle}
\newcommand{\raa}{\rangle\hspace{-0.08cm}\rangle}
\newcommand{\supV}{\mathds{1}_{V_1}}

\newcommand{\LZ}{L^2(\mathbb{R}^3,\mathbb{C})}

\newcommand{\LZN}{L^2(\mathbb{R}^{3N},\mathbb{C})}

\newcommand{\im}{\text{i}}

\newtheorem{theorem}{Theorem}[section]
\newtheorem{lemma}[theorem]{Lemma}

\newtheorem{corollary}[theorem]  {Corollary}
\newtheorem{remark}[theorem]  {Remark}
\newtheorem{definition}[theorem] {Definition}
\newtheorem{assumption}[theorem] {Assumption}
\newtheorem{proposition}[theorem]{Proposition}

\newenvironment{proof}{\emph{Proof:}}{\begin{flushright} $ \Box $ \end{flushright}}
\newenvironment{refproof}{\emph{Proof of Theorem }}

\renewcommand{\phi}{\varphi}

\newcommand{\beq}{\begin{equation}}
\newcommand{\eeq}{\end{equation}}

\begin{document}

\title{Derivation of the time dependent Gross-Pitaevskii equation for a class of non purely positive potentials }

%\titlerunning{Derivation of the Gross-Pitaevskii equation}

\author{Maximilian Jeblick\footnote{
pickl@math.lmu.de
Mathematisches Institut, LMU Munich, Germany} $ $ and Peter Pickl\footnote{
pickl@math.lmu.de
Mathematisches Institut, LMU Munich, Germany}}

%

%\date{Received: date / Revised version: date}

\maketitle

\begin{abstract}
We present a microscopic derivation of the time-dependent Gross-Pitaevskii equation
 starting from an interacting $N$-particle system of Bosons. 
We prove convergence of the reduced density matrix corresponding to the exact time evolution to the projector onto the solution of the respective Gross-Pitaevskii equation.
 Our work extends a previous result by one of us (P.P. \cite{picklgp3d}) to interaction potentials 
 which need not to be nonnegative, but
 may have a sufficiently small negative part. One key estimate in our proof is an operator inequality 
which was first proven by Jun Yin, see \cite{yin}.

\end{abstract}
\newpage
\tableofcontents\newpage
\section{Introduction}

The main concern of this work is  a generalization of a previous result presented by one of us (P.P. \cite{picklgp3d}).
Specifically, we will analyze the dynamics of a Bose-Einstein condensate in the
Gross-Pitaevskii regime for interactions $V$ which need not to be nonnegative, but may have an attractive part. 

Let us first define the $N$-body quantum problem we want to study. 
The evolution of $N$ interacting bosons is described by a time-dependent wave-function $\Psi_t \in L^2_{s}(\mathbb{R}^{3N}, \mathbb{C}),
\| \Psi_t \|=1$ (throughout this paper norms without index $\|\cdot\|$ always denote the $L^2$-norm on the appropriate Hilbert space.).
The bosonic $N$-particle Hilbert space
$L^2_{s}(\mathbb{R}^{3N}, \mathbb{C})$ denotes the set of all $\Psi \in  L^2(\mathbb{R}^{3N}, \mathbb{C})$
 which are symmetric under pairwise permutations of the variables 
$x_1, \dots, x_N \in \mathbb{R}^3$. Assuming in addition $\Psi_0 \in H^2 (\mathbb{R}^{3N},\mathbb{C})$, the evolution of $\Psi_t$ is then described by the  $N$-particle 
Schr\"odinger equation
\begin{align}
\label{schroe}
 \im\partial_t \Psi_t =  H\Psi_t 
 \;.
\end{align}
The time-dependent Hamiltonian $ H$ we will study is defined by
\begin{align}
\label{hamiltonian}
 H=-\sum_{j=1}^N \Delta_j+
 N^2
 \sum_{1\leq  j< k\leq  N}V(N (x_j-x_k)) +\sum _{j=1}^N A_t(x_j)
 \; .
\end{align} 
In the following, we assume $A_t \in L^\infty(\mathbb{R}^ 3, \mathbb{R})$ 
and
 $V \in L_c^\infty (\mathbb{R}^ 3, \mathbb{R})$, $V$ spherically symmetric.
We will also use the common notation $V_1(x)=N^2 V(Nx)$. 
More generally, one can study the properties of Bose gases for
a larger class of
scaling parameters $0 \leq \beta \leq 1$, setting
$V_\beta(x)= N^{-1+3 \beta} V(N^\beta x)$. For $0<\beta \leq 1$ and large particle number $N$, the potential gets $\delta$-like, which indicates that the mathematical description may become more involved the bigger $\beta$ is chosen.
The so-called Gross-Pitaevskii regime $\beta=1$ is special, since then the two-particle correlations play a crucial role for the dynamics, see Section \ref{secmic}.

We will derive an approximate solution of \eqref{schroe} in the trace class topology of reduced density matrices.
Define the one particle reduced density matrix $\gamma^{(1)}_{\Psi_0}$ given by the integral kernel
\begin{align}
\gamma^{(1)}_{\Psi_0}(x,x')=\int_{\mathbb{R}^{3N-3}} \Psi_{0}^*(x,x_2,\ldots,x_N)\Psi_{0}(x',x_2,\ldots,x_N)d^3x_2\ldots d^3x_N \;. 
\end{align}
To account for the physical situation of a Bose-Einstein condensate, we assume complete condensation in the limit of large particle number $N$. This amounts to
assume that, for $N \rightarrow \infty$,
 $\gamma^{(1)}_{\Psi_0}  \rightarrow |\phi_0\rangle\langle\phi_0|$ in trace norm for some $\phi_0
\in L^2(\mathbb{R}^3,\mathbb{C}) 
  , \|\phi_0\|=1$.
  Our main goal is to show the persistence of condensation over time.
  Let $a$ denote the scattering length of the potential $\frac{1}{2} V$ (see Section
\ref{secmic} for the precise definition of $a$) and
   let $\phi_t$ solve  the nonlinear Gross-Pitaevskii equation
\begin{align}
\label{GP}
\im \partial_t
\phi_t=\left(-\Delta +A_t\right) \phi_t+
8 \pi a
|\phi_t|^2\phi_t=:h^{\text{GP}}\phi_t
\end{align}
with initial datum $\phi_0$ (we assume $\phi_t \in H^2(\mathbb{R}^3, \mathbb{C})$, see below). 
We then prove that the time evolved reduced density matrix $\gamma^{(1)}_{\Psi_t}$
 converges to $ |\phi_t\rangle\langle\phi_t|$ in trace norm as $N \rightarrow \infty$ with convergence rate of order $N^{-\eta}$ for some $\eta>0$. 
%For $V \in L^\infty_c(\mathbb{R}^3, \mathbb{R})$, such that $V$ is spherically symmetric and  $V \geq 0$, the convergence $\gamma^{(1)}_{\Psi_t} \rightarrow| \phi_t \rangle \langle \phi_t |$ in trace norm is well known, under suitable conditions on
%$\Psi_0$ and $\phi_0$, see \cite{benedikter, brennecke, erdos3, picklgp3d}.

The rigorous derivation of effective evolution equations has a long history, see e.g. \cite{benedikter, brennecke,fluctnbody, SchleinNorm, chen2d, erdos1, erdos2, erdos4, erdos3, jeblick, jeblick2, keler, teufel, schlein2d, knowles, michelangeli2, michelangeli3, picklnorm, focusingnorm,  marcin1, marcin2, picklgp3d, pickl1, rodnianskischlein}
 and references therein.
 The derivation of the three dimensional time-dependent Gross-Pitaevskii equation for nonnegative potentials was first conducted in \cite{erdos3}.
 Afterward, this result has been improved by \cite{benedikter, brennecke, michelangeli2, picklgp3d}.
In the two dimensional case, the correspondent time-dependent Gross-Pitaevskii equation
was treated in \cite{jeblick}. Note that in two dimensions, the scaling considered is given by $e^{2N} V(e^Nx)$.
The ground state properties of dilute Bose gases were treated in \cite{boccato, gpbog, lewin1, lewin2, lewin3, ls, lsy, rougerie, pizzo, yin}, see also the monograph \cite{lssy} and references therein.

As mentioned previously, we will generalize the result presented by one of us (P.P. \cite{picklgp3d}) to a specific class of interactions $V$ which are not assumed to be nonnegative everywhere.
Let us stress that persistence of condensation is not expected for arbitrary $V$.
For strongly attractive potentials, even a small fraction of particles which leave the condensate over time may cluster, subsequently causing the condensate to collapse in finite time.
The dynamical collapse of a Bose gas under such circumstances is well known within the physical community and was mathematically treated in \cite{michelangeli3}.
The breakdown of condensation has also been observed in experiments \cite{collapse}.
Consequently, the result we are going to prove can only be valid under certain
restrictions on  $V$. 
The class of potentials we consider is chosen such that $V$ has a repulsive core, i.e. there exists a $r_1>0$, such that $V(x)\geq\lambda^+$, for some $\lambda^ +>0$ and for all $|x| \leq r_1$. This condition prevents clustering of particles.
If furthermore the negative part of $V$ fulfills some restrictions (see assumption \ref{Vassumption}), a result by Jun Yin \cite{yin} then implies
that the Hamiltonian we consider in this note is stable of second kind.
The author proves in particular that for such potentials the ground state energy per particle of a dilute, homogeneous Bose gas is
at first order given by the well-known formula $ 4 \pi a \rho N$.
Among the steps of the proof in \cite{yin}, it is shown that the
Hamiltonian \eqref{hamiltonian} -without external potential $A_t$- restricted to configurations where at least three particles are close to each other is a nonnegative operator. 
We will adapt this non-trivial operator inequality in our proof to 
control the kinetic energy of those particles which leave the condensate, see Lemma \ref{energylemma}. 
%Consequently, this note will essentially be concerned with the class of potentials defined in \cite{yin}, see assumption \ref{Vassumption}.
We like to remark that the assumptions \ref{Vassumption} on $V$ stated below 
imply that the scattering length $a$ of the potential $\frac{1}{2} V$ is nonnegative. Consequently, the effective Gross-Pitaevskii  dynamics \eqref{GP} is repulsive, which reflects the fact that the condensate is stable.

The result presented in \cite{yin} implies further that there exists an $\epsilon>0$, such that
\begin{align}
\label{operatorboundT}
-\epsilon \sum_{k=1}^N \Delta_k \leq -\sum_{k=1}^N \Delta_k+ \sum_{i<j} V_1(x_i-x_j),
\\
\label{operatorboundV}
\epsilon \sum_{i<j} |V_1(x_i-x_j)| \leq -\sum_{k=1}^N \Delta_k+ \sum_{i<j} V_1(x_i-x_j).
\end{align}
The first operator inequality bounds $\|\nabla_1 \Psi_t\| $
uniformly in $N$, if initially
the energy per particle is of order $1$. If this were not the case, 
one cannot expect condensation, see e.g. \cite{michelangeli3} for a nice discussion.
Under the same assumption, the second inequality \eqref{operatorboundV} implies 
$
\|V_1(x_1-x_2) \Psi_t \| \leq N^ {1/2}
$, see Lemma \ref{corest}. These two inequalities are crucial in our proof to control the rate of particles which leave the condensate over time
and thus to extend the result presented in \cite{picklgp3d}.
%It is also important to note that
%the operator inequality \eqref{operatorboundT} can only hold, if the scattering length $a$ fulfills $a \geq 0$, see \cite{seiringer} and the discussion below. Thus, the effective Gross-Pitaevskii dynamics of the condensate is repulsive and also reflects the fact that the condensate is stable.

%-------------------------

\section{Main Result}
\label{secmain}

We will bound  expressions which are uniformly bounded in $N$ by some (possible time-dependent) constant $C>0$.
We will not distinguish constants
 appearing in a sequence of estimates, i.e. in $X\leq  CY\leq  CZ$ the constants usually differ.
 We denote by  $\laa\cdot,\cdot\raa$  the scalar product on $\LZN$
and by $\langle\cdot,\cdot\rangle$  the scalar product on $\LZ$.
We will use the notation $B_r(x) =
\lbrace z \in \mathbb{R}^3| |x-z| < r \rbrace$.

Define the energy functional $\mathcal{E}:H^2(\mathbb{R}^{3N}, \mathbb{C})\to \mathbb{R}$
\begin{align}
\mathcal{E}(\Psi)=N^{-1}\laa\Psi,H\Psi\raa ,
\end{align}
as well as the
Gross-Pitaevskii energy functional $\mathcal{E}^{GP}:H^2(\mathbb{R}^3, \mathbb{C})\to \mathbb{R}$  \begin{align}\label{energyfunct}
\mathcal{E}^{GP}(\phi):=&\langle\nabla\phi,\nabla\phi\rangle+\langle\phi,(A_t+4 \pi a|\phi|^2)\phi\rangle
=\langle\phi,
(h^{GP}-4 \pi a|\phi|^2)\phi\rangle .
\end{align}
Next, we will define the class of interaction potentials $V$ we will consider. This class is essentially the one considered in \cite{yin}, Theorem 2; see also Corollary 1 and Corollary 2 in \cite{yin} for a different characterization of the class of potentials $V$. In this note, we require in addition
that the potential changes its sign only once. This facilitates the discussion of the scattering state, see Section \ref{secmic}. In principle, one could ease this additional assumption by generalizing the proofs given in Section \ref{secmic}
\begin{definition}
\label{cubedef}
Divide $\mathbb{R}^3$ into cubes $C_n ,\;n \in \mathbb{Z}$ of side length $b_1/\sqrt{3}$; that
is $\mathbb{R}^3 = \cup_{n =- \infty}^\infty C_n$. Furthermore, assume that $\mathring{C_n}\cap \mathring{C_m}= \emptyset$ for $m \neq n$.
. Define
\begin{align*}
n(b_1,b_2)=
\max _{x \in \mathbb{R}^3}
\# \{ n: C_n \cap B_{b_2}(x) \neq \emptyset \} .
\end{align*}
Thus, $n(b_1,b_2)$ gives the maximal number of of cubes with side length $b_1/\sqrt{3}$ one needs to cover a sphere with radius $b_2$.
We remark that
$
4 \sqrt{3} \pi (\frac{b_2}{b_1}-1)^3 \leq
n(b_1,b_2)
\leq
\frac{\frac{4 \pi}{3}(b_1+b_2)^3}{b_1^3 3^{-3/2}}
=
4 \sqrt{3} \pi (1+ \frac{b_2}{b_1})^3
$.

\end{definition}

\begin{assumption}
\label{Vassumption}
Let $V \in L_c^\infty( \mathbb{R}^3, \mathbb{R})$ spherically symmetric and
let $V(x)=V^+(x)-V^-(x)$, where
$V^+,V^- \in L_c^\infty( \mathbb{R}^3, \mathbb{R})$ are spherically symmetric, such that $V^+(x),V^-(x) \geq 0$ and the supports of $V^+$ and $V^-$ are disjoint.
Assume that
\begin{enumerate}
\item
For $R> r_2>0$, we have
$\text{supp}(V^+)=B_{r_2}(0)$ and
$\text{supp}(V^-)=B_{R}(0)\setminus B_{r_2}(0)$.
\item
There exists $\lambda^+>0$ and $r_1>0$, such that 
$V^+(x) \geq \lambda^+$ for all $ x \in B_{r_1}(0)$.
\item
Define $\lambda^- = \|V^-\|_\infty$ as well as 
$n_1=n(r_1,R) \text{ and } n_2=n(r_1,3R).$
%We remark that
%$
%n_2 \leq n_1
%\left(
%\frac{r_1+3 R}{R-r_1}
%\right)^3
%$.
Define, for $0<\epsilon<1$,
\begin{align}
\label{minenergyscatt}
\mathcal{E}_R(\phi)
=
\int_{B_R(0)}
\left(
|\nabla_x \phi(x) |^2
+
\frac{1}{1-\epsilon}
n_1
(2 V^+(x)-4 V^-(x)) 
|\phi(x)|^2
\right)d^3x.
\end{align}
We then assume that for some $0<\epsilon <1$
\begin{align}
\label{scattposcondition}
&\inf_{\phi \in C^1(\mathbb{R}^3, \mathbb{C}),  \phi(R)=1}
\left( \mathcal{E}_R(\phi)\right)
\geq 0,
\\
&\lambda^+  >  8 n_2  \lambda^- .
\end{align} 
\end{enumerate}

\end{assumption}
\begin{remark}
We will use the constants $r_1,r_2,R$, $\lambda^+, \lambda^-$, as well as $n_1,n_2$ throughout this paper as defined above.
\end{remark}

\begin{remark}
Condition \eqref{scattposcondition} implies $a \geq 0$, see
Theorem C.1.,(C.8.) in \cite{lssy}.
Assumption \ref{Vassumption} implies that there exists $\epsilon>0, \mu>0$ such that
\begin{align}
\label{Hpos}
&-\sum_{k=1}^N
\Delta_k
+
\sum_{i<j=1}^N
(V_1^+(x_i-x_j)
-
(1+ \epsilon)
V_1^-(x_i-x_j)
)
\geq 
0,
\\
&-(1- \mu)\sum_{k=1}^N
\Delta_k
+
\sum_{i<j=1}^N
V_1(x_i-x_j)
\geq 
0 ,
\end{align}
see Lemma \ref{HposLemma} and Corollary \ref{Hcor}.
The operator inequality \eqref{Hpos} can only hold for $a \geq 0$, see \cite{seiringer}  and is thus in
accordance with Condition \eqref{scattposcondition}.
Thus, although the potential $V$ may have an attractive part $V^-$, the effective Gross-Pitaevskii equation \eqref{GP} is repulsive.

It also follows from assumption \ref{Vassumption} (c)
\begin{align}
-\Delta+ \frac{1}{2}V \geq 0.
\end{align}
%In two dimensions (but not in three), condition \eqref{scattposcondition} is equivalent to
%\begin{align}
%-\Delta_x+
%\frac{1}{1-\epsilon}
%n_1
%(2 V^+(x)-4 (1+ \mu) V^-(x)) \geq 0,
%\end{align}
%see Theorem C.1.,(C.8.) in \cite{lssy}.

\end{remark}
We now state the main Theorem:
\begin{theorem}\label{theo}
Let $\Psi_0 \in L^2_{s}(\mathbb{R}^{3N}, \mathbb{C}) \cap H^2(\mathbb{R}^{3N}, \mathbb{C})$ with $\|\Psi_0\|=1$. Let $\phi_0  \in H^{2}(\mathbb{R}^{3},\mathbb{C})$ with $\|\phi_0\|=1$. Let
$\lim\limits_{N\to\infty}
\text{Tr} | \gamma^{(1)}_{\Psi_0}-|\phi_0\rangle\langle\phi_0| |
 =0 $, as well as  $\lim\limits_{N\to\infty}\mathcal{E}(\Psi_0)=\mathcal{E}^{GP}(\phi_0)$.
Let $\Psi_t$ the unique solution to $i \partial_t \Psi_t
=  H \Psi_t$ with initial datum $\Psi_0$ and assume that $V$ fulfills assumption \ref{Vassumption}.
Let $\phi_t$ the unique solution to $i \partial_t \phi_t
= h^{\text{GP}} \phi_t$ with initial datum $\phi_0$ and assume $\phi_t \in H^2(\mathbb{R}^{3}, \mathbb{C})$. Let the external potential $A_t$ fulfill $A_t, \dot{A}_t \in L^\infty(\mathbb{R}^{3}, \mathbb{R})$ for all $t \in \mathbb{R}$.

Then,
\begin{itemize}
\item[(a)] for any $t>0$ \begin{equation}\label{converge}\lim_{N\to\infty}\mu_1^{\Psi_t}=|\phi_t\rangle\langle\phi_t|\end{equation} in operator norm.
\item[(b)]  if $\int_0^\infty (\|\phi_s\|_\infty+\|\nabla\phi_s\|_{6,loc}+\|\dot{A}_s\|_\infty) ds<\infty$ where  $\|\cdot\|_{6,loc}:\LZ\to\mathbb{R}^+$ is the ``local $L^6$-norm'' given by $$\|\phi\|_{6,loc}:=\sup_{x\in\mathbb{R}^3}\|\mathds{1}_{|\cdot-x|\leq  1} \phi\|_6\;,$$ then
the convergence (\ref{converge}) is uniform in $t>0$.
\end{itemize}
\end{theorem}
\begin{remark}
\begin{enumerate}
\item Note that convergence of $\mu_1^{\Psi}$ to $|\phi\rangle\langle\phi|$ in operator norm is equivalent to convergence in trace norm, since $|\phi\rangle\langle\phi|$ is a rank
one projection \cite{rodnianskischlein}. 
Other equivalent definitions of asymptotic 100\% condensation can be found in \cite{michelangeli}.

\item 
%In case of a repulsive potential $V \geq 0$
%Lieb, Seiringer and Yngvason have shown that in the limit $N\to\infty$ the energy-difference $\mathcal{E}(\Psi^{gs})-\mathcal{E}^{GP}(\phi^{gs})%\to0$, where $\Psi^{gs}$ is the ground state of a trapped Bose gas and $\phi^{gs}$  the ground state of the
%respective Gross-Pitaevskii energy functional \cite{lsy}. In \cite{ls} Lieb and Seiringer show that
%$\mu^{\Psi^{gs}}\to|\phi^{gs}\rangle\langle\phi^{gs}|$.
For potentials $V$ which satisfy assumption $\ref{Vassumption}$, convergence of 
$\mathcal{E}(\Psi^{gs})-\mathcal{E}^{GP}(\phi^{gs})\to0$ was shown in \cite{yin} for homogeneous gases.

\item By Sobolev's inequality, it follows that $\|\nabla\phi_s\|_{6,loc}\leq \|\nabla\phi_s\|_{6}\leq \|\Delta\phi\|$. Thus $\|\nabla\phi_s\|_{6,loc}$ can be bounded controlling $\langle\phi_s, \left(h^{GP}\right)^2\phi_s\rangle$ sufficiently well.

     On the other hand, $\|\nabla\phi_s\|_{6,loc}\leq \|\nabla\phi_s\|_\infty$. Since we are in the defocussing regime one expects, after the potential is turned off,  that $\|\phi\|_\infty$ and $\|\nabla\phi\|_\infty$ decay like $t^{-3/2}$. Whenever this is the case $\int_0^\infty \|\phi_s\|_\infty+\|\nabla\phi_s\|_{6,loc}+\|\dot{A}_s\|_\infty ds<\infty$ and we get convergence uniformly in $t$.

\item
Existence of solutions of the Gross-Pitaevskii equation is well understood.
The condition $\phi_t \in H^2(\mathbb{R}^ 3,\mathbb{C})$ can be proven for a large class of external
potentials, assuming sufficient regularity of the initial datum $\phi_0$, see e.g. \cite{cazenave}.

\item The proof of Theorem \ref{theo} implies that the rate of convergence is of order
$N^{-\delta}$ for some $\delta>0$, assuming that 
$| \gamma^{(1)}_{\Psi_0}-|\phi_0\rangle\langle\phi_0| | \leq C N^{-2 \delta}$, as well as assuming that the convergence rate of
$\lim\limits_{N\to\infty}\mathcal{E}(\Psi_0)=\mathcal{E}^{GP}(\phi_0)$ to be least of order $N^{-2\delta}$. 

\item
The Theorem can straightforwardly be adapted to the two-dimensional case. There, one
considers the scaling $V_N(x)= e^{2N} V(e^N x)$, for $V \in L^\infty_c(\mathbb{R}^2, \mathbb{R})$
spherically symmetric, see \cite{jeblick}. 
Note that due to the different scaling behavior of the potential, most of the respective bounds given below read differently in two dimensions.
In this note, we are mainly concerned with the three-dimensional case. However, we will also give the respective proofs of certain Lemmata for the two-dimensional system in cases where some nontrivial modifications are needed. 

\end{enumerate}
\end{remark}

\section{Proof of Theorem \ref{theo}}\label{secpro}

The method our proof relies on is explained in details in  \cite{pickl1}. Heuristically speaking it is based on the idea of counting for each time $t$ the relative number of those particles which are
not in the state $\phi_t$  and estimating the time derivative of that value. 
In this note we will only focus on the modifications one needs to perform in order to generalize the result of \cite{picklgp3d} to more general interactions $V$. 
We will therefore often omit large parts of existing proofs and refer  the reader to 
\cite{picklgp3d} for the detailed steps and motivations.

First, we will recall some important definitions we will need during the proof.

\begin{definition}\label{defpro}
Let $\phi\in\LZ$.
\begin{enumerate}
\item For any $1\leq  j\leq  N$ the
projectors $p_j^\phi:\LZN\to\LZN$ and $q_j^\phi:\LZN\to\LZN$ are defined by
\begin{align*} p_j^\phi\Psi=\phi(x_j)\int\phi^*(\tilde{x}_j)
\Psi(x_1,\ldots,\tilde{x}_j,\ldots,x_N)
d^3\tilde{x}_j\;\;\;\forall\;\Psi\in\LZN
\end{align*}
and $q_j^\phi=1-p_j^\phi$.

We  also use the bra-ket notation
$p_j^\phi=|\phi(x_j)\rangle\langle\phi(x_j)|$. For better readability, we will sometimes
use the notation $p_j,q_j$.
\item 
For any $0\leq  k\leq  N$ we define the set $$\mathcal{S}_k:=\{(s_1,s_2,\ldots,s_N)\in\{0,1\}^N\;;\;
\sum_{j=1}^N s_j=k\}$$ and the orthogonal projector $P_{k}^\phi$
acting on $\LZN$ as
$$P_{k}^\phi:=\sum_{\vec a\in\mathcal{S}_k}\prod_{j=1}^N\big(p_{j}^{\phi}\big)^{1-s_j} \big(q_{j}^{\phi}\big)^{s_j}\;.$$
For
negative $k$ and $k>N$ we set $P_{k}^\phi:=0$.
\item
For any function $m:\mathbb{N}_0\to\mathbb{R}^+_0$ we define the operator $\widehat{m}^{\phi}:\LZN\to\LZN$ as
\be\label{hut}\widehat{m}^{\phi}:=\sum_{j=0}^N m(j)P_j^\phi\;.\ee
We furthermore define $\widehat{n}^\phi$ with $n(k)= \sqrt{\frac{k}{N}}$.
\end{enumerate}

\end{definition}

\begin{definition}\label{hdetail}
For any $1 \leq j \neq k \leq N$, let \be\label{defkleins}a_{j,k}:=\{(x_1,x_2,\ldots,x_N)\in
\mathbb{R}^{3N}: |x_j-x_k|<N^{-26/27}\},\ee
\be
\overline{\mathcal{A}}_j:=\bigcup_{k\neq j}a_{j,k}\;\;\;\;\;\;\;\mathcal{A}_j:=\mathbb{R}^{3N}\backslash \overline{\mathcal{A}}_j
\;\;\;\;\;\;\;\overline{\mathcal{B}}_{j}:=\bigcup_{k,l\neq j}a_{k,l}\;\;\;\;\;\;\;\mathcal{B}_{j}:=\mathbb{R}^{3N}\backslash
\overline{\mathcal{B}}_{j}\;.
\ee
(In two dimensions, the sets $\mathcal{A}_j$ and $\mathcal{B}_{j}$ are defined differently, see
\cite{jeblick}.)
Furthermore, define for any set $A \subset \mathbb{R}^{3N}$ the operator
$
\mathds{1}_{A}: L^2(\mathbb{R}^{3N}, \mathbb{C})
\rightarrow
L^2(\mathbb{R}^{3N}, \mathbb{C})
$
 as the projection onto the set $A$.
\end{definition}

Many Lemmata which were proven in  \cite{picklgp3d} are valid for generic interaction potentials $V$ and need not to be modified. 
In the following , we will state a general criteria under which assumptions on $\Psi_t$ Theorem \ref{theo}
is valid (see (b),(c) and (d) below). Subsequently, we prove that these assumptions are valid if the potential $V$ fulfills assumption
\ref{Vassumption}.

\begin{lemma}
\label{lemmafuertheo}
Let $\Psi_0 \in L^2_{s}(\mathbb{R}^{3N}, \mathbb{C}) \cap H^2(\mathbb{R}^{3N}, \mathbb{C})$ with $\|\Psi_0\|=1$. 
Let $\phi_0  \in H^{2}(\mathbb{R}^3,\mathbb{C})$ with $\|\phi_0\|=1$.
Let  $\lim\limits_{N\to\infty}\gamma^{(1)}_{\Psi_0}=|\phi_0\rangle\langle\phi_0|$ in trace norm
as well as $
\lim\limits_{N \rightarrow \infty}
\mathcal{E}(\Psi_0)
=
\mathcal{E}^ {GP}(\phi_0)
$.
Let $\Psi_t$ the unique solution to $i \partial_t \Psi_t
= H \Psi_t$ with initial datum $\Psi_0$ and assume $V \in L_c^\infty(\mathbb{R}^3, \mathbb{R})$ spherically symmetric.
Let $\phi_t$ the unique solution to $i \partial_t \phi_t
= h^{GP} \phi_t$ with initial datum $\phi_0$.
Assume $A_t, \dot{A}_t \in L^\infty (\mathbb{R}^3, \mathbb{R})$.
If,
\begin{enumerate}
\item
\begin{align}
\label{condition1}
\phi_t \in H^ 2(\mathbb{R}^3, \mathbb{C}).
\end{align}
\item
\begin{align}
\label{condition2}
\|V_1(x_1-x_2) \Psi_t \| \leq C N^ {1/2}.
\end{align}
\item
\begin{align}
\label{condition3}
\|\nabla_1 \Psi_t \| \leq C.
\end{align}
\item
for some $\eta>0$, the following inequality holds:
\begin{align}
\label{condition4}
\|\mathds{1}_{\mathcal{A}_{1}}\nabla_1q^{\phi_t}_1
\Psi_t \|^2 
+
\|\mathds{1}_{\overline{\mathcal{B}}_{1}}\nabla_1
\Psi_t \|^2
\leq  C \left(\llaa\Psi_t,\widehat{n}^{\phi_t}\Psi_t\rraa+N^{-\eta}\right)+\left|\mathcal{E}(\Psi_t)-\mathcal{E}^{GP}(\phi_t)\right|.
 \end{align}
\item 
\begin{align}
\label{condition5}
\text{V is chosen such that Lemma }\ref{defAlemma}\text{ is fulfilled.}
\end{align}
\end{enumerate}
Then, for any $t>0$ 
\begin{equation}
\label{convergenls}
\lim_{N\to\infty}\gamma^{(1)}_{\Psi_t}=|\phi_t\rangle\langle\phi_t|
\end{equation} in trace norm.

\end{lemma}
\begin{remark}
It has been shown in \cite{jeblick, picklgp3d} that the conditions \eqref{condition2},
\eqref{condition3},\eqref{condition4} and \eqref{condition5} are fulfilled
for nonnegative potentials $V \in L^\infty_c(\mathbb{R}^{d}, \mathbb{R})$, with $d=2,3$\footnote{
Condition
\eqref{condition2} reads $\| e^{2N} V(e^N(x_1-x_2) )\Psi_t\| \leq C e^N N^{-1/2}$ in two dimensions, see Lemma 7.8 in \cite{jeblick}. 
Furthermore, for the two-dimensional system, we need higher regularity of $\phi_t$. There, 
condition \eqref{condition1} reads $\phi_t \in H^3 (\mathbb{R}^2, \mathbb{C})$.}.
Conditions \eqref{condition2}-\eqref{condition4} are essentially
those conditions which are non-trivial to prove 
and also lead to the class of potentials \ref{Vassumption} we consider in this note.

We furthermore like to remark that our proof of Lemma \ref{defAlemma} uses the assumption that $V$ changes its sign only once and that $V$ is positive 
around the origin. As mentioned, we expect Lemma
\ref{defAlemma} to be valid for a larger class of potentials than those defined in assumption \ref{Vassumption}.
\end{remark}

\begin{proof}
We like to recall the scheme of the proof
of the equivalent of Theorem \ref{theo} for nonnegative
potentials. The proof presented in 
 \cite{picklgp3d} can be seen as a two-step argument.
First, it is shown in Section 6.2.2. in 
 \cite{picklgp3d} that the convergence \eqref{convergenls}
 \textit{generally} follows, if 
 certain functionals
$
\gamma_x(\Psi_t,\phi_t)
$, with $ x \in \{a,b,c,d,e,f\}$,
can be bounded sufficiently well, that is
$
|\gamma_x (\Psi_t, \phi_t)| \leq C N^{-\delta}$, $\delta>0$ \footnote{The
functional called
 $\gamma_f( \Psi_t, \phi_t)$ is actually missing in \cite{picklgp3d}. The definition of this functional can be found in
  in equation (6.10) \cite{michelangeli2} and in p. 32 \cite{jeblick}. In these papers, it is furthermore shown that
  the respective bound
  $
|\gamma_f (\Psi_t, \phi_t)| \leq C N^{-\delta}$, $\delta>0$  holds, assuming $V \in L^\infty_c( \mathbb{R}^d, \mathbb{R})$ to be nonnegative. 
}.
The exact definition of these functionals can be found in Definition 6.2. and Definition 6.3. in \cite{picklgp3d}.

It is then proven in Lemma A.1. in  \cite{picklgp3d}  that the bound  
$
|\gamma_x (\Psi_t, \phi_t)| \leq C N^{-\delta}$, $ x \in \{a,b,c,d,e\}$ is valid for
nonnegative potentials $V \in L_c^\infty(\mathbb{R}^3, \mathbb{C})$.

In the following, we will show that the estimates
$
|\gamma_x (\Psi_t, \phi_t)| \leq C N^{-\delta}$
given in \cite{picklgp3d} 
remain valid under the
conditions \eqref{condition1}-\eqref{condition5}. 
Note that we will not restate the estimates given in \cite{picklgp3d}, but only focus on the modifications one needs to perform. 
\medskip\\
The bound of $|\gamma_a(\Psi_t, \phi_t)| \leq C N^{-\delta}$ directly follows from $\dot{A}_t\in L^\infty (\mathbb{R}^3, \mathbb{R})$, see Lemma A.1. in \cite{picklgp3d}.
The required bound of $|\gamma_b(\Psi_t, \phi_t)|$ is derived in Lemma A.4., pp.31-37 in \cite{picklgp3d}.
Following the estimates given in \cite{picklgp3d},
it can be verified line-by-line that the given bounds are valid, if conditions \eqref{condition1}-\eqref{condition5} and $A_t \in L^\infty(\mathbb{R}^3, \mathbb{R})$ hold.
Furthermore, it can be verified that 
the functionals  $\gamma_c$ and $\gamma_e$can be controlled using 
conditions \eqref{condition1}-\eqref{condition5}, see Lemma A.1. and pp.38-42 in \cite{picklgp3d}.
The estimate for  $\gamma_f$ is valid under conditions \eqref{condition1} and \eqref{condition5} and can be found in 
 p. 34 in \cite{michelangeli2} and 
p. 53 in \cite{jeblick}.
\medskip \\
In two dimensions, $\gamma_d$ can be bounded, using conditions
\eqref{condition1}-\eqref{condition5}, 
see pp.50-52 \cite{jeblick} (we like to recall that the $N$-dependent bounds given in Lemma 
\ref{defAlemma}
read slightly different in two dimensions).
\medskip\\
In three dimensions,
the functional $\gamma_d$ can be bounded, using in addition the following estimate:
Let
$
m^a(k)
=m(k)-m(k+1)
$, where, for some $\xi>0$, 
$$
m(k)
=
\begin{cases}
\sqrt{k/N},  &\text{ for } k \geq N^{1-2 \xi}, \\
1/2(N^{-1+ \xi}k+ N^{-\xi}), &\text{else.}
\end{cases}
$$
We control
\begin{align}
\label{qq2Neu}
N^3\left|\laa\Psi_t ,\mathds{1}_{\overline{\mathcal{B}}_{1}}g_{\bet,1}(x_{1}-x_{3})V_1(x_1-x_2)
\widehat{m^a}^{\phi_t} p^{\phi_t}_{1}\mathds{1}_{\mathcal{B}_{3}}\Psi_t\raa\right|,
\end{align}
where $g_{\beta_1,1}$ is defined in Lemma \ref{defAlemma}.
This term, which appears in (A.49) in \cite{picklgp3d} is the only term in $\gamma_x(\Psi_t, \phi_t)$, $x \in \lbrace
a,b,c,d,e,f  \rbrace$ where the estimate given in \cite{picklgp3d} needs to be modified,
using only the assumptions given in the Lemma above.
By a general inequality (see Lemma 4.3. in \cite{picklgp3d} and (A.50)-(A.52) in \cite{picklgp3d}), it can be verified that
\begin{align}
(\ref{qq2Neu})
\leq&N^{-1- \epsilon}\|\mathds{1}_{\overline{\mathcal{B}}_{1}}V_1(x_1-x_2)\Psi_t\|^2\label{csym0neu}
\\&+
CN^{6+ \epsilon}\|g_{\bet,1}(x_{1}-x_{3})\supV(x_1-x_2)
\widehat{m^a}^{\phi_t}  p^{\phi_t}_{1}\mathds{1}_{\mathcal{B}_{3}}\Psi_t\|^2\label{csym1neu}
\\&+CN^{7+ \epsilon} \big|\laa\Psi_t,\mathds{1}_{\mathcal{B}_{3}}p^{\phi_t}_1 \widehat{m^a}^{\phi_t}  g_{\bet,1}(x_{1}-x_{3})\nonumber
\\&\hspace{3cm}\supV(x_1-x_2)g_{\bet,1}(x_{1}-x_{4})
\widehat{m^a}^{\phi_t}  p^{\phi_t}_{1}\mathds{1}_{\mathcal{B}_{4}}\Psi_t\raa\big|\label{csym2neu}
\;
\end{align}
for all $\epsilon \in \mathbb{R}$.
For nonegative $V$ and $\epsilon=0$, it was possible to control
(\ref{csym0neu}) using a specific energy estimate, see Lemma 5.2.(3) in \cite{picklgp3d}.
We do not expect this estimate to hold for potentials $V$ which are not nonnegative.
 For an interaction potential $V$, fulfilling condition \eqref{condition2}, we can however bound
\begin{align*}
\eqref{csym0neu}
\leq
C
N^{-\epsilon}.
\end{align*}
The estimate
$(\ref{csym1neu})\leq CN^{-1+2\xi+ \epsilon}$ given in (A.51) \cite{picklgp3d} is valid under conditions \eqref{condition1}-\eqref{condition5}.
Note that condition \eqref{condition5} implies
$\|g_{\beta_1,1}(x_1-x_2) \Omega \|
\leq C \| \nabla_1 \Omega \|$ for $\Omega \in L^2(\mathbb{R}^{3N},\mathbb{C})$, see
Lemma \ref{defAlemma}. This is one key estimate in order to bound $(\ref{csym1neu})$.
Under the some conditions, it has been shown (c.f. (A.52) in \cite{picklgp3d}) that
\begin{align*}
(\ref{csym2neu})
\leq&CN^{-\frac{26}{9}+3\xi+ \epsilon}
\;.
\end{align*}
Therefore, it follows for some $\eta>0$ that
\begin{align}
(\ref{qq2Neu})
\leq
C
 N^{-\eta}
\end{align}
holds by choosing $\xi>0$ and $\epsilon>0$ small enough\footnote{
 Note that the factors
$N^{2 \xi}$ and $N^{3\xi} $ are due to the definition of $m(k)$.
A factor of the form $N^{s \xi}$, $s \in \lbrace1,2,3 \rbrace$ also appears in
the other functionals
$\gamma_x(\Psi_t, \phi_t)$, $x \in \lbrace
b,c,e, f \rbrace$. It therefore follows that the
the respective bounds 
$|\gamma_x(\Psi_t, \phi_t)| \leq C N^{-\eta}, \eta>0$
given in \cite{picklgp3d} are valid choosing $\xi >0$  small enough. 
We like to remark that one cannot choose $ \xi=0$, since the convergence of the reduced density matrices stated in Lemma \ref{lemmafuertheo} 
does only follow for $0<\xi <1/2$, see \cite{picklgp3d} for the precise argument. 
}.

\end{proof}

\begin{refproof}\ref{theo}:
In the following, we will prove the inequalities \eqref{condition2}, \eqref{condition3} and \eqref{condition4} for interaction potentials which fulfill assumption \ref{Vassumption}.
Theorem \ref{theo}, part (a) then follows from Lemma \ref{lemmafuertheo}, together with the estimates given 
in Section \ref{secmic}.
Part (b) of Theorem \ref{theo} follows from part (a) and the estimates given in 
\cite{picklgp3d}.
 \begin{flushright} $ \Box $ \end{flushright}
\end{refproof}

\subsection{The scattering state}
\label{secmic}
In this section we analyze the microscopic structure which is induced by $V_1$. 
While the principle estimates are the same as in \cite{jeblick, picklgp3d}, we need to modify
the proofs given there which relied on the nonnegativity of $V$.

\begin{definition}
 Let $V \in L^\infty_c(\mathbb{R}^3, \mathbb{R})$ fulfill assumption \ref{Vassumption}.
Define the zero energy scattering state $j$ by
\begin{equation}
\label{eq: defj}
\begin{cases} 
\left( - \Delta_x + \frac{1}{2} V( x)  \right) j(x)=0 ,
\\
\lim\limits_{ |x| \rightarrow \infty} j(x)=1.
  \end{cases}
\end{equation}
Furthermore define the scattering length $a$ by
\begin{align}
a=
\text{scat} \left( \frac{1}{2}V \right)
= 
\frac{1}{4 \pi}
\int \frac{1}{2} V(x) j(x) d^3x .
\end{align}
\end{definition}
We want to recall some important properties of the scattering state $j$, see also Appendix C of \cite{lssy}.

\begin{lemma} \label{scattlemma fuer v}
For the scattering state defined previously the following relations hold:
\begin{enumerate}
\item
$j$ is a nonnegative, monotone nondecreasing function which is spherically symmetric in $|x|$. For $|x| \geq R$, $j$ is given by
 $$
 j (x)
= 1 - \frac{a}{|x|}. $$
\item
The scattering length $a$ fulfills $a \geq 0$. 
\end{enumerate}
\end{lemma}

\begin{proof}
\begin{enumerate}
\item[(a)+(b)] 

Since we assume 
$-\Delta+ \frac{1}{2}V \geq 0$,
one can define the scattering state $j$ by a variational principle. Theorem C.1 in \cite{lssy} then implies that $j$ is a nonnegative, spherically symmetric function in $|x|$ such that
 $
 j (x)
= 1 - \frac{a}{|x|} $ holds for $|x| \geq R$
with $a \in \mathbb{R}$ defined as above. 
By condition \eqref{scattposcondition} it follows $a \geq 0$, see Theorem C.1., (C.8.) in \cite{lssy}.
It is only left to show that $j$ is monotone nondecreasing in $|x|$.
Let
$t(|x|)=j(x)$ and define
\begin{align*}
a_r
=
\frac{1}{4 \pi}
 \int_0^r \frac{1}{2} V(r' e_{r'}) t(r') (r')^2 dr' ,
\end{align*}
where $e_{r'}$ denotes the radial unit vector.
Note that $a= \lim_{r \rightarrow \infty} a_r= a_R$.
By Gau{\ss}-theorem and the scattering equation \eqref{eq: defj}, it then follows for $r>0$
\begin{align*}
\frac{d}{dr}t(r)
=
\frac{a_r}{r^2} .
\end{align*}
Since $t(r) \geq 0$ holds for all $r \geq 0$, it follows 
$a_r>0$ for all $r \in ]0,r_2[$.
If it were now that $j$ is not monotone nondecreasing, there must exist a $\tilde{r} \geq r_2$, such that $a_{\tilde{r}}<0$.
$V(x) \leq 0$ and $t(r) \geq 0$ for all $|x| \in ]r_2,R[$ then imply $a_r \leq a_{\tilde{r}}$ for all
$r \geq \tilde{r}$. This, however, contradicts $a=a_R \geq0$. Thus, it follows that $j$ is monotone nondecreasing.

\end{enumerate}
\end{proof}

Using a general idea, we will define a potential $W_{\beta_1}$ with $0<\beta_1<1$, such that
$\frac{1}{2}(V_1-W_{\beta_1})$ has scattering length zero. 
This allows us to ``replace'' $V_1$  by $W_{\beta_1}$, which has better scaling behavior and is easier to control.

\begin{definition}\label{microscopic}
Let $V \in L_c^\infty( \mathbb{R}^3, \mathbb{R})$ satisfy assumption \ref{Vassumption}.
Let $a_N$ denote the scattering length of $\frac{1}{2} V_1(x)=\frac{1}{2} N^2V(Nx)$.
For any $0<\beta_1<1$ and any $R_{\beta_1} \geq N^{- \beta_1}$ we define the potential $W_{\beta_1}$ via
\begin{align}
\label{eq: defW}
W_{\beta_1}(x)
=
 \begin{cases} 
a_N N^{3 \beta_1} & \text{if } N^{- \beta_1} < |x| \leq R_{\beta_1},  \\
   0      & \text{else} .
  \end{cases}
\end{align}
Furthermore, we define the zero energy scattering state $f_{\beta_1,1}$ of the potential
$ \frac{1}{2} (V_1-W_{\beta_1})$, that is
\begin{align}
\label{eq: deff}
 \begin{cases} 
\left( - \Delta_x + \frac{1}{2} \left(V_1(x)-W_{\beta_1}(x) \right)  \right) f_{\beta_1,1}(x)=0,
\\
f_{\beta_1,1}(x)=1 \; \text{for } |x| = R_{\beta_1} .
  \end{cases}
\end{align}
\end{definition}

\begin{remark}
Note, by scaling, that $a_N= N^ {-1} a$. Furthermore $j_N(x):=j(Nx)$ solves
\begin{equation*}
\begin{cases} 
\left( - \Delta_x + \frac{1}{2} V_1( x) 
 \right) j_N(x)=0 ,
\\
\lim\limits_{ |x| \rightarrow \infty} j_N(x)=1 .
  \end{cases}
\end{equation*}
\end{remark}

In the following Lemma we show that there exists a minimal value $R_{\beta_1}$ such that the scattering length of the potential 
$\frac{1}{2}(V_1-W_{\beta_1})$ is zero. 
%The vanishing of the scattering length is equivalent to the condition $\int  (V_1(x)-W_{\beta_1}(x)) f_{\beta_1,1}(x) d^3x =0$. 
In the two dimensional case, the analog of Lemma \ref{defAlemma} is, except part (i), also valid in two dimensions (one has to replace the bounds below by the respective bounds given in \cite{jeblick}. Furthermore, $W_{\beta_1}$ is defined differently.). Part (i) needs not to be proven in two dimensions, see Remark \ref{iremark}.

\begin{lemma}\label{defAlemma}
For the scattering state $f_{\beta_1,1}$, defined by \eqref{eq: deff}, the following relations hold :
\begin{enumerate}
\item[(a)]
 There exists a minimal value $R_{\beta_1} \in \mathbb{R}$ such that $\int  (V_1(x)-W_{\beta_1}(x)) f_{\beta_1,1}(x) d^3x=0$. 
 \end{enumerate}
For the rest of the paper we assume that $R_{\beta_1}$ is chosen such that (a) holds.
  \begin{enumerate}
\item[(b)]
There exists $K_{\beta_1} \in \mathbb{R}, \; K_{\beta_1}> 0$ such that
$K_{\beta_1} f_{\beta_1,1}(x) = j(N x) \; \forall |x| \leq N^{-\beta_1}$.

\item[(c)]
 $f_{\beta_1,1}$ is a nonnegative, monotone nondecreasing function in $|x|$. Furthermore,
\begin{align}
f_{\beta_1,1}(x)=1 \; \text{for } |x| \geq R_{\beta_1} \;.
\end{align}
\item[(d)]
\begin{align}
1 \geq K_{\beta_1} \ge
1-\frac{a}{ N^{1-\beta_1}}
\; .
\end{align}
\item[(e)]
$ R_{\beta_1} \leq C N^{-\beta_1}$.
\end{enumerate}
For any fixed $0<\beta_1$, $N$ sufficiently large such that $V_1$ and $W_{\beta_1}$ do not overlap, we obtain
\begin{enumerate}
\item[(f)]
\begin{align*}
&| N \| V_{1}f_{\beta_1,1} \|_1 - 8 \pi  a | 
=
| N \| W_{\beta_1} f_{\beta_1,1} \|_1 - 8 \pi  a | \leq C N^{-1-\beta_1}
  \;.
\end{align*}

\item[(g)]
Define 
\begin{align*}
g_{\beta_1,1}(x) = 1 - f_{\beta_1,1}(x)
\;.
\end{align*}
Then,
\begin{align*}
\|g_{\beta_1,1}\|_1&\leq  C N^{-1-2\beta_1}
\;,\hspace{0.1cm}
\|g_{\beta_1,1}\|_{3/2}&\leq  C N^{-1-\beta_1}
\;,\hspace{0.1cm}
\|g_{\beta_1,1}\|\leq  C N^{-1-\beta_1/2}
\;,\hspace{0.1cm}
\|g_{\beta_1,1}\|_\infty \leq 1
 \;.
\end{align*}
\item[(h)]
\begin{align*}
| N \| W_{\beta_1}\|_1 - 8 \pi  a| \leq C N^{-1+\beta_1}
  \;.
\end{align*}
\item[(i)]
For any $\Omega \in H^1(\mathbb{R}^{3N}, \mathbb{C})$, we have
$$
\|g_{\beta_1,1}(x_1-x_2) \Omega \| \leq CN^{-1} \| \nabla_1 \Omega\|.
$$

%\item[(i)]
%\begin{align*}
%W_{\beta_1} \in \mathcal{V}_{\beta_1}
%\;,
%W_{\beta} f_{\beta_1,1} \in \mathcal{V}_{\beta_1}
%\;.
%\end{align*}
\end{enumerate}

\end{lemma}

\begin{proof}
\begin{enumerate}
\item
In the following, we will sometimes denote, 
with a slight abuse of notation, $f_{\beta_1, 1}(x)=f_{\beta_1, 1}(r)$
and $j(x)=j(r)$
 for $r=|x|$ (for this, recall that $f_{\beta_1, 1}$ and $j$ are radially symmetric). 
 We further denote by $f'_{\beta_1, 1}(r)$ the derivative of $f_{\beta_1, 1}$ with respect to
 the radial coordinate $r$.
We first show by contradiction that $f_{\beta_1, 1}(N^{-\beta_1}) \neq 0$.
For this, assume  that 
$f_{\beta_1, 1} (x) =0$ for all $|x| \leq N^{-\beta_1}$. 
Since $f_{\beta_1, 1}$ is continuous, there exists a maximal value $r_0 \geq N^{-\beta_1}$ such that the scattering
equation \eqref{eq: deff} is equivalent to
\begin{align}
\label{nosolution}
 \begin{cases} 
\left( - \Delta_x - \frac{1}{2} W_{\beta_1}(x) \right) f_{\beta_1, 1}(x)=0,
\\
f_{\beta_1, 1}(x)=1 \; \text{for } |x| = R_{\beta_1},
\\
f_{\beta_1,1}(x)=0 \; \text{for } |x| \leq r_0
\;.
  \end{cases}
\end{align}
Using \eqref{eq: deff} and Gauss'-theorem, we further obtain
\begin{align}
\label{integral v-w scattstate}
f'_{\beta_1, 1}(r) =
\frac{1}{ 8 \pi r^2 }
\int_{B_r(0)} d^3 x (V_1(x)-W_{\beta_1}(x)) f_{\beta_1, 1}(x) 
\;.
\end{align}
 \eqref{nosolution} and \eqref{integral v-w scattstate} then imply for $r > r_0$
 \begin{align*}
&
 \left| f'_{\beta_1, 1}(r)
\right| 
  =
\frac{1}{ 8 \pi r^2 }
\left|
\int_{B_r(0)} d^3 x W_{\beta_1}(x) f_{\beta_1, 1}(x) 
\right|
=
\frac{a 
N^{-1+3 \beta_1}
}{ 2 r^2 }
\left|
\int_{r_0}^r dr' r'^2f_{\beta_1, 1}(r') 
\right|
\\
\leq&
\frac{a 
N^{-1+3 \beta_1}
}{ 2 r^2 }
\left|
\int_{r_0}^r dr' r'^2
(r'-r_0)  
\sup_{r_0 \leq s \leq r}
|f'_{\beta_1, 1}(s)|
\right| .
 \end{align*}
Taking the supreme over the interval 
$[r_0,r]$, the inequality above then implies that
there exists a constant $C(r,r_0) \neq 0$, $ \lim\limits_{r \rightarrow r_0} C(r,r_0)=0$ such that
 $
 \sup\limits_{r_0 \leq s \leq r}
|f'_{\beta_1, 1}(s)|
 \leq
C(r,r_0)
N^{-1+3 \beta_1}
\sup\limits_{r_0 \leq s \leq r}
|f'_{\beta_1, 1}(s)|
 $. Thus, for $r$ close enough to $r_0$, the inequality above can only hold if 
 $
f'_{\beta_1, 1}(s)=0 
 $ for $s \in [r_0,r]$, yielding a contradiction to the choice of $r_0$.
\\
Consequently, there exists a $x_0 \in \mathbb{R}^3, |x_0| \leq N^{-\beta_1}$, such that
$f_{\beta_1, 1}(x_0) \neq 0$.
%Since, for $|x| \leq r_1$, $V(x) \geq \lambda^+>0$ is positive, 
%Theorem C.1, together with Lemma C.2 A) in \cite{lssy} imply
%there exists a $x_0, 0<|x_0| \leq N^{-\beta_1}$ such that $f_{\beta_1, 1} (x_0) \neq 0$. 
We can thus define 
\begin{equation*}
h(x) = f_{\beta_1, 1}(x)\frac{j(N x_0)}{f_{\beta_1, 1}(x_0)}
\end{equation*}
on the compact set $\overline{B_{x_0}(0)}$.
One easily sees  that $h(x)= j(N x)$ on $\partial \overline{B_{x_0}(0)}$ and satisfies the zero energy scattering equation~\eqref{eq: defj} for $ x \in \overline{B_{N^{-\beta_1}}(0)}$.
Note that the scattering equations \eqref{eq: defj} and \eqref{eq: deff} have a unique solution on any compact set.
It then follows that 
$h(x) = j(N x) \; \forall x \in \overline{B_{N^{-\beta_1}}(0)}$.
Since $j(N N^{-\beta_1}) \neq 0$, we then obtain
$f_{\beta_1, 1}(N^{-\beta_1}) \neq 0$.

Thus, $f_{\beta_1, 1}(x)=j(Nx)\frac{f_{\beta_1, 1}(x_0)}{j(N x_0)}$ holds for all $|x| \leq N^{-\beta_1}$ and for all $x_0 \in ]0, N^{-\beta_1}]$.
 Lemma \ref{scattlemma fuer v} further implies
that 
 either $f_{\beta_1, 1}$ or $-f_{\beta_1, 1}$ is a nonnegative, spherically symmetric and monotone nondecreasing function in $|x|$
 for all $|x| \leq N^{-\beta_1}$.

Recall that $W_{\beta_1}$ and hence $f_{\beta_1, 1}(x) $ depend on $R_{\beta_1}
\in [N^{-\beta_1}, \infty[
$. 
 For conceptual clarity, we
 denote $W_{\beta_1}^{(R_{\beta_1})}(x) = W_{\beta_1}(x) $ and
$f_{\beta_1, 1}^{(R_{\beta_1})}(x)= f_{\beta_1, 1}(x)$ for the rest of the proof of part (a).
For $\beta_1$ fixed, consider the function
\begin{align}
& s: [N^{-\beta_1} ,\infty [ 
\rightarrow  \mathbb{R} \\
&R_{\beta_1}
\mapsto  
\int_{B_{R_{\beta_1}}(0)} d^3 x
 (V_1(x)-W^{(R_{\beta_1})}_{\beta_1}(x)) f^{(R_{\beta_1})}_{\beta_1, 1} (x).
\end{align}
%Assume w.l.o.g. $f^{(R_{\beta_1})}_{\beta_1, 1}(x) \geq 0$ for $|x|\leq N^{-\beta_1}$.
We show  by contradiction that the function $s$ has at least one zero.
Assume $s \neq 0$ were to hold.
 We can assume w.l.o.g. $s >0$. 
It  then
follows from Gauss'-theorem that 
$
f'^{(R_{\beta_1})}_{\beta_1, 1}(R_{\beta_1})> 0
$ for all  
$R_\beta \geq N^{-\beta_1}$.
By uniqueness of the solution of the scattering equation \eqref{eq: deff}, for $\tilde{R}_{\beta_1}<R_{\beta_1}$ there exists a constant
$ K_{\tilde{R}_{\beta_1},R_{\beta_1}} \neq 0$, such that for all
$
|x|\leq \tilde{R}_{\beta_1}
$
we have
$
f^{(\tilde{R}_{\beta_1})}_{\beta_1, 1}(x)= K_{\tilde{R}_{\beta_1},R_{{\beta}_1}}f^{(R_{\beta_1})}_{\beta_1, 1}(x)
$. 
If $K_{\tilde{R}_{\beta_1},R_{{\beta}_1}} < 0$ were to hold, we could conclude from
\begin{align*}
0<&
s(\tilde{R}_{\beta_1})= 8 \pi (\tilde{R}_{\beta_1})^2  f'^{(\tilde{R}_{\beta_1})}_{\beta_1, 1}(\tilde{R}_{\beta_1})
=
8 \pi (\tilde{R}_{\beta_1})^2  
K_{\tilde{R}_{\beta_1},R_{{\beta}_1}}f'^{(R_{\beta_1})}_{\beta_1, 1}(
\tilde{R}_{\beta_1})
\end{align*}
that $f'^{(R_{\beta_1})}_{\beta_1, 1}(
\tilde{R}_{\beta_1}) < 0$. By continuity of 
$f'^{(R_{\beta_1})}_{\beta_1, 1}$ and
$
f'^{(R_{\beta_1})}_{\beta_1, 1}(R_{\beta_1})> 0
$, there exists $ r \in ]\tilde{R}_{\beta_1}, R_{\beta_1}[$, such that $
0=f'^{(R_{\beta_1})}_{\beta_1, 1}(r)=
K_{R_{\beta_1},r}
f'^{(r)}_{\beta_1, 1}(r)
$
, yielding to a contradiction to $s > 0$.
\\
We can therefore conclude
 $ K_{\tilde{R}_{\beta_1},R_{{\beta}_1}}> 0$.
From Lemma \ref{scattlemma fuer v}, the assumption $s(N^{-\beta_1})>0$ and  $ K_{\tilde{R}_{\beta_1},R_{{\beta}_1}} > 0$, we obtain,
 for all $r \in [0, N^{-\beta_1} ] $ and for all $R_{\beta_1} \in [N^{-\beta_1} ,\infty [ $, that
$f^{(R_{\beta_1})}_{\beta_1, 1} (r) \geq 0$ holds. 
 From $ s \neq 0$, it then follows that, for all $r \in [N^{-\beta_1} ,\infty [  $ and for all $R_{\beta_1} \in [N^{-\beta_1} ,\infty [ $ , $
f'^{(R_{\beta_1})}_{\beta_1, 1}(r) \neq 0 $.
Thus, for all $r \in [N^{-\beta_1} ,\infty [  $ and for all $R_{\beta_1} \in [N^{-\beta_1} ,\infty [ $, the function
$f^{(R_{\beta_1})}_{\beta_1, 1}(r)$ doesn't change sign. 
This, however, implies $\lim\limits_{R_{\beta_1} \rightarrow \infty} s(R_{\beta_1})= - \infty$ yielding to a contradiction.  
By continuity of $s$, there exists thus a minimal value $R_{\beta_1} \geq N^{-\beta_1} $ such that
$s(R_{\beta_1})=0$.

\begin{remark}
As mentioned, we will from now on fix $R_{\beta_1} \in [N^{- \beta_1}, \infty[$ as the minimal value such that $s(R_{\beta_1})=0$.
Furthermore, we may assume $a>0$ and $R_{\beta_1} >N^{-\beta_1}$ in the following. For $a=0$, we can choose $R_{\beta_1}=N^{-\beta_1}$, such that $f_{\beta_1,1}(x)=
j(N x)/j(N N^{-\beta_1})$. It is then easy to verify that the Lemma stated is valid.
\end{remark}

\item
From
$j(N x) = f_{\beta_1, 1}(x)\frac{j(N N^{-\beta_1})}{f_{\beta_1, 1}(N^{-\beta_1})}$, for all $|x| \leq N^{-\beta_1}$,
we can conclude that
\begin{equation}
\label{defKbeta}
K_{\beta_1} = \frac{j(N N^{-\beta_1})}{f_{\beta_1, 1}(N^{-\beta_1})}.
\end{equation}
Next, we show that the constant $K_{\beta_1}$ is positive.
Since $j(N N^{-\beta_1})$ is positive, it follows from Eq.~\eqref{defKbeta} that  $K_{\beta_1}$ and $f_{\beta_1, 1} (N^ {-\beta_1})$ have equal sign. By (a), the sign of $f_{\beta_1, 1}$ is constant for $|x|\leq R_{\beta_1}$.
Furthermore, from Gauss'-theorem and the scattering equation \eqref{eq: deff} we have
\begin{align}
 \label{Gauss}
f'_{\beta_1, 1}(r)=
\frac{1}{8 \pi  r^2 K_{\beta_1}}
\int_{B_r(0)} V_1(x) j(N x) d^3x
\end{align}
for all $0<r \leq N^{- \beta_1}$. Since 
$\int_{B_r(0)} V_1(x) j(N x) d^3x$ is nonnegative for all $0<r \leq N^{- \beta_1}$ (see the proof of Lemma 
\ref{scattlemma fuer v}), we then conclude
\begin{align}
\label{eq: sign of derivative of f}
\text{sgn} \left(
f'_{\beta_1, 1}(N^{-\beta_1})
\right) = \text{sgn}(K_{\beta_1}).
\end{align}
Recall that
  $
f'_{\beta_1, 1}(R_\beta)=0  
$.
If it were now that $K_{\beta_1}$ is negative, we could conclude from \eqref{defKbeta} and \eqref{eq: sign of derivative of f} that
$
f'_{\beta_1, 1}(N^{-\beta_1})<0$ and
$ f_{\beta_1, 1} ( N^{- \beta_1}) <0$. 
Since $R_{\beta_1}$ is by definition the smallest value where $f'_{\beta_1, 1}$ vanishes, we were able to conclude from the continuity of the derivative that  $f'_{\beta_1, 1}(r)<0$ for all $r < R_{\beta_1}$ and hence $f(R_{\beta_1}) <0$. However, this were in contradiction to the boundary condition of the zero energy scattering state (see \eqref{eq: deff}) and thus $K_{\beta_1} > 0$ follows.

\item
From the proof of property (b), we see that $f_{\beta_1, 1}$ and its derivative is positive at $N^{- \beta_1}$. 
From \eqref{integral v-w scattstate}, we obtain $f'_{\beta_1, 1}(r) = 0 $ for all  $r > R_{\beta_1}$. Thus $f_{\beta_1, 1}(x) =1$ for all $|x| \geq R_{\beta_1}$.
 Due to continuity $ f'_{\beta_1,1}(r)  >0 $ for all $r < R_{\beta_1}$. Since $f_{\beta_1, 1}$ is continuous, positive at  $N^{- \beta_1}$, and its derivative is a nonnegative function, it follows that $f_{\beta_1, 1}$ is a nonnegative, monotone nondecreasing function in $|x|$.

\item
Since $f_{\beta_1, 1}$ is a positive monotone nondecreasing function in $|x|$, we obtain
\begin{align*}
1 \geq f_{\beta_1, 1}( N^{- {\beta_1}}) =j(N N^{- \beta_1}) /K_{\beta_1}
=
\left(
1- \frac{a}{N^{1- \beta_1}} 
\right)
/ K_{\beta_1}.
\end{align*}
We obtain the lower bound
\begin{align*}
K_{\beta_1} \ge
1- \frac{a}{N^{1- \beta_1}} . 
\end{align*}

For the upper bound, we first prove that $f_{\beta}(x) \geq j(Nx)/j(NR_{\beta_1})$ holds for all $|x | \leq N^{-\beta_1}$.
Define $m(x)= j(N x)/j(NR_{\beta_1})-f_{\beta_1, 1}(x)$.
Using the scatting equations \eqref{eq: defj} and \eqref{eq: deff}, we obtain
\begin{align}
\label{differenzstreu}
\begin{cases}
\Delta_x
m(x)
=
\frac{1}{2}
V_1(x)
m(x)
+
\frac{1}{2}
W_{\beta_1}(x) f_{\beta_1, 1}(x),
\\
m(R_{\beta_1})=0.
\end{cases}
\end{align}
Since $W_{\beta_1}(x) f_{\beta_1, 1}(x) \geq 0$, we obtain that
$\Delta_x
m(x) \geq 0$ for $  N^{-1} R \leq |x|\leq R_{\beta_1}$. That is, $m(x)$ is subharmonic
for $ N^{-1} R < |x|< R_{\beta_1}$. Using the maximum principle, we obtain, using that 
$m(x)$ is spherically symmetric
\begin{align}
\label{minimum}
\max_{N^{-1} R\leq |x|\leq R_{\beta_1}} (m(x))
=
\max_{|x| \in \lbrace N^{-1} R ,  R_{\beta_1} \rbrace} ( m(x))
\;.
\end{align}
If it were now that $\max_{|x| \in \lbrace N^{-1} R ,  R_{\beta_1} \rbrace} (m(x))
=m(N^{-1} R)
\geq
m(R_\beta)=0
$, we  could assume $ m(x) > 0$ for all
$ N^{-1} R \leq |x|\leq N^{-\beta_1}$
(otherwise we would have $m (N^{- \beta_1})=0$, which implies
$
K_{\beta_1}
=
j(NR_{\beta_1})
=
1- \frac{a}{NR_{\beta_1}} \leq 
1
$).
Note that $m(x)$ then solves
\begin{align*}
\begin{cases}
-\Delta_x
m(x)
+
\frac{1}{2}
V_1(x)
m(x)=0
\; \text{ for } |x| \leq N^{-\beta_1},
\\
m(N^{-1}R) > 0 .
\end{cases}
\end{align*}
By  Theorem C.1 in \cite{lssy} (note that we can assume $a>0$),
$m$ is strictly increasing for $N^{-1} R \leq |x| \leq N^{-\beta_1}$. 
This, however, contradicts $\max_{|x| \in \lbrace N^{-1}R ,  R_{\beta_1} \rbrace} (m(x))
=m(N^{-1} R)$.

Therefore, we can conclude in \eqref{minimum} that
$\max_{|x| \in \lbrace N^{-1} R ,  R_{\beta_1} \rbrace} (m(x))
=m(R_{\beta_1})=0$
 holds.
  Then, it follows that $f_\beta(x)-j_{N,R_\beta}(x) \geq 0$ for all $N^{-1} R \leq |x| \leq N^{-\beta_1}$. Using the zero energy scattering equation
$$
-\Delta 
(
f_{\beta_1,1}(x)-j(Nx)/j(NR_{\beta_1})
)
+
\frac{1}{2}
V_1(x)
(
f_{\beta_1,1}(x)-j(Nx)/j(NR_{\beta_1})
)
=
0
$$ for $|x| \leq N^{-\beta_1}$, we can, together with $  f_{\beta_1,1}(N^{-\beta_1})-
j(NN^{-\beta_1})/j(NR_{\beta_1}) \geq 0$, conclude that  $ f_{\beta_1,1}(x)- j(Nx)/j(NR_{\beta_1}) \geq 0$ for all $|x| \leq R_{\beta_1}$. 

As a consequence, we obtain the desired bound
$
K_\beta=
\frac{j(NN^{-\beta_1})}{f_{\beta_1,1}(N^{-\beta_1})}\leq j(NR_{\beta_1}) \leq 1$.

\item
Since $f_{\beta_1, 1}$ is a nonnegative, monotone nondecreasing function in $|x|$, it follows that
\begin{align*}
 N^{-1}  f_{\beta_1, 1}(N^{-\beta_1}) 
\int V(x)  d^3x 
=&
f_{\beta_1, 1}(N^{-\beta_1}) 
\int V_1(x)  d^3x
\geq
\int V_1(x) f_{\beta_1, 1}(x) d^3x
\\
=&
\int W_{\beta_1}(x) f_{\beta_1, 1}(x) d^3x
\geq
f_{\beta_1, 1} (N^{-\beta_1})
\int W_{\beta_1} (x) d^3x
\;.
\end{align*}
Therefore, $\int W_{\beta_1} (x) d^3x \leq C N^{-1}$ holds, which implies that $R_{\beta_1} \leq C N^{-\beta_1}$.

\begin{remark}
We will now prove the
the two-dimensional analog of (e) which requires a more refined estimate.
This is due to the fact that
$
\int_{\mathbb{R}^2} e^{2N} V(e^N x) d^2x= \mathcal{O}(1)$ does not decay like $N^{-1}$.
We refer to \cite{jeblick} for the precise definition and notation we use in the following.
\end{remark}
\textit{Proof of part (e) for the two-dimensional system:}
\\
Since $f_{\beta}$ is a nonnegative, monotone nondecreasing function in $|x|$ with $f_{\beta} (x)=1$ $\forall |x| \geq R_{\beta}$, it follows that
\begin{align*}
 f_{\beta}(N^{-\beta})  \int_{\mathbb{R}^2}d^2x V(x)  
=&
f_{\beta}(N^{-\beta}) 
\int_{\mathbb{R}^2}d^2x V_N(x)  
\geq
\int V_N(x) f_{\beta}(x) d^2x
\\
=&
\int_{\mathbb{R}^2}d^2x W_{\beta}(x) f_{\beta}(x) 
\geq
f_{\beta} (N^{-\beta})
\int_{\mathbb{R}^2}d^2x W_{\beta} (x) 
\;.
\end{align*}
Therefore, $\int_{\mathbb{R}^2}d^2x W_{\beta} (x) \leq C$ holds, which implies that $R_{\beta} \leq C N^{1/2-\beta}$.
From
\begin{align*}
\frac{1}{K_\beta} \frac{4 \pi}{N+ \ln \left(\frac{R_{\beta}}{a}\right)}
=&
\frac{1}{K_\beta}
\int_{\mathbb{R}^2}d^2x V_N(x) j_{N,R_{\beta}}(x)
=
\int_{\mathbb{R}^2}d^2x V_N(x) f_{\beta}(x)
\\
=&
\int_{\mathbb{R}^2}d^2x M_\beta(x) f_{\beta}(x)
= 
8 \pi^2 N^{-1+2 \beta} 
\int_{N^{-\beta_1}}^{R_{\beta}} 
dr r f_{\beta}(r)
\end{align*}
we conclude that
\begin{align*}
\int_{N^{-\beta}}^{R_{\beta}} 
dr r f_{\beta}(r)
=
\frac{N^{1- 2 \beta}}{2 \pi K_\beta
\left(
N+ \ln \left(\frac{R_{\beta}}{a} \right)
\right)
} 
\;.
\end{align*}
Since $f_{\beta}$ is a nonegative, monotone nondecreasing function in $|x|$, 
\begin{align*}
\frac{1}{2}( R_{\beta}^2- N^{-2 \beta})
\frac{j_{N,R_{\beta}}(N^{-\beta})}{K_\beta}
=
\frac{1}{2}( R_{\beta}^2- N^{-2 \beta})
f_{\beta}(N^{-\beta})
\leq
\int_{N^{-\beta}}^{R_{\beta}} 
dr r f_{\beta}(r)
\end{align*}
which implies
\begin{align*}
R_{\beta}^2 N^{2 \beta}
\leq
\frac{N}{\pi 
 \left(N+ \ln \left(\frac{R_{\beta}}{a} \right) \right)
j_{N,R_{\beta}}(N^{-\beta})}+1.
\end{align*}
Using $R_{\beta}\leq C N^{1/2- \beta}$, it then follows
\begin{align*}
 j_{N,R_{\beta}}(N^{-\beta})
 =
 1+   \frac{1}{N+ \ln \left(\frac{R_{\beta}}{a}\right)} \ln \left( \frac{N^{-\beta}}{R_{\beta}} \right)
 \geq 1- \frac{C}{N}
 \;,
\end{align*}
which implies $R_{\beta} \leq C N^{-\beta}$.

\item
Using
\begin{align*}
\| W_{\beta_1} f_{\beta_1, 1}\|_1=&
\| V_1 f_{\beta_1, 1} \|_1 = K_{\beta_1}^ {-1} \|V_1 j (N \cdot) \|_1
=
K_{\beta_1}^{-1} 8 \pi \frac{a}{N}
\;,
\end{align*}
we obtain
\begin{align*}
|N\| V_1 f_{\beta_1, 1}\|_1- 8\pi a|
=&
|N\| W_{\beta_1} f_{\beta_1, 1}\|_1- 8\pi a|
=
8 \pi
\left| K_{\beta_1}^{-1}-1 \right|\leq C N^{-1+\beta_1}
\;.
\end{align*}

\item
Using for $|x| \leq R_{\beta_1}$ the inequality
 $1\ge f_{\beta_1, 1}(x) \geq j(Nx)/j(NR_{\beta_1})$,  it follows for $|x| \leq R_{\beta_1}$ 
\begin{align*}
0\leq& g_{\beta_1,1}(x) = 1- f_{\beta_1,1}(x) \le 1- j(Nx)/j(NR_{\beta_1})
 \;.
\end{align*}
Let $\tilde{j}$ solve
\begin{align*}
\begin{cases} 
\left( - \Delta_x + \frac{1}{2} V( x) \mathds{1}_{|x| \leq r_2} \right) \tilde{j}(x)=0 ,
\\
\tilde{j}(2R)=j(2R) .
  \end{cases}
\end{align*}
It then follows that $\tilde{a}= \text{scat} \left( \frac{1}{2}V( x) \mathds{1}_{|x| \leq r_2} \right)>0$. 
Furthermore, it follows from Theorem C.1 and Lemma C.2 in \cite{lssy} that
$$
\tilde{j}(x) \geq 
\frac{
1- \frac{\tilde{a}}{|x|}}
{1- \frac{\tilde{a}}{2R}} j(2R)
=
\left(1- \frac{\tilde{a}}{|x|}\right)
\frac{
1- \frac{a}{2R}}
{1- \frac{\tilde{a}}{2R}} 
$$ holds for all $x \in \mathbb{R}^3$.
Consider $n(x)= \tilde{j}(x)-j(x)$. $n$ then solves
\begin{align*}
\begin{cases} 
\Delta_x n(x) =
 \frac{1}{2} V( x) n(x)
+
 \frac{1}{2} V( x) \mathds{1}_{|x| \leq r_2} \tilde{j}(x),
\\
n(2R)=0 .
  \end{cases}
\end{align*}
As before (see \eqref{differenzstreu}), we can conclude $n(x) \leq 0$ for all $|x| \leq 2R$, which implies
$ j(x) \geq \tilde{j}(x)$, for $|x| \leq 2R$.
Therefore,
\begin{align*}
j(Nx) \geq
\begin{cases}
\left(1- \frac{\tilde{a}}{N|x|}\right)
\frac{
1- \frac{a}{2NR}}
{1- \frac{\tilde{a}}{2NR}}
\text{ for } N|x| \leq R ,
\\
1- \frac{a}{N|x|}
\text{ else}.
\end{cases}
\end{align*}
This implies, using part (d),
\begin{align}
\nonumber
g_{\beta_1,1}(x)
\leq &
1-
\begin{cases}
\left(1- \frac{\tilde{a}}{N|x|}\right)
\frac{
1- \frac{a}{2NR}}
{(1- \frac{\tilde{a}}{2NR})
(1- \frac{a}{NR_{\beta_1}})
}
\text{ for } N|x| \leq R ,
\\
\frac{
1- \frac{a}{N|x|}
}
{(1- \frac{a}{NR_{\beta_1}})}
\text{ else}.
\end{cases}
\\
\leq
\label{gbetaest}
&
\begin{cases}
\frac{\tilde{a}}{N|x|}
+
C N^{-1}
\text{ for } N|x| \leq R ,
\\
\frac{a}{N|x|}
+
 C N^{-1+ \beta_1}
\text{ else}.
\end{cases}
\end{align}

Since $g_{\beta_1,1}(x)=0$ for $\vert x\vert > R_{\beta}$, we conclude with $R_{\beta_1} \leq C N^{-\beta_1}$
that
\begin{align*}
\|g_{\beta_1,1}\|_1 
\leq &  
N^{-1 -2 \beta_1},
\end{align*}
as well as
\begin{align*}
\| g_{\beta_1,1}\|_{3/2}\leq
C N^{-1- \beta_1}, \qquad
\| g_{\beta_1,1}\|\leq
C N^{-1- \beta_1/2}.
\end{align*}
Furthermore,
$ \| g_{\beta_1,1}\|_\infty = \|1- f_{\beta_1,1}\|_\infty \leq 1$, since $f_{\beta_1, 1}$ is a nonnegative, monotone nondecreasing function with $f_{\beta_1, 1}(x) \leq 1$.

\item 
Using (f) and (g), we obtain with $\| W_{\beta_1}\|_1 \leq CN^{-1}$
\begin{align*}
&| N \| W_{\beta_1}\|_1 - 8 \pi a | 
\leq
| N \|W_{\beta_1}f_{\beta_1, 1}\|_1 - 8 \pi a| 
+
 N \| W_{\beta_1}g_{\beta_1,1}\|_1 
 \\
& \leq
 C
\left( 
 N^{-1+\beta_1} +
   \|\mathds{1}_{|\cdot| \geq N^{-\beta_1}} g_{\beta_1,1}\|_\infty
  \right)
  \;.
\end{align*}
Since $g_{\beta_1,1}(x)$ is a nonnegative, monotone nonincreasing function, it follows with $K_{\beta_1} \leq 1$
\begin{align*}
 &  \|\mathds{1}_{|\cdot| \geq N^{-\beta_1}} g_{\beta_1,1}\|_\infty
   =
   g_{\beta_1,1}(N^{-\beta_1})=1-f_{\beta_1, 1}(N^{-\beta_1})
   =
   1- \frac{j(NN^{-\beta_1})}{K_{\beta_1}}
   \leq
a N^{-1+\beta_1}
  \;.
\end{align*}
and (h) follows.
\item
Using the pointwise estimate \eqref{gbetaest}, we obtain
for any $\Omega \in H^1(\mathbb{R}^{3N}, \mathbb{C})$
\begin{align*}
\|g_{\beta_1,1}(x_1-x_2) \Omega \| \leq 
C( N^{-1+\beta_1} 
\|\mathds{1}_{B_{CN^{-\beta_1}}(0)}(x_1-x_2)\Omega\|+ 
N^{-1} \| 
|x_1-x_2|^{-1}
 \Omega\|).
\end{align*}

Since $\| |x_1-x_2|^{-1}
 \Omega\|\leq 2\|\nabla_1 \Omega\|$
as well as
$\|\mathds{1}_{B_{CN^{-\beta_1}}(0)}(x_1-x_2)\Omega\|
 \leq C N^{-3\beta_1/2}\|\nabla_1 \Omega\|$
  holds, we obtain part (i).
\begin{remark}
\label{iremark}
Part (i) is not valid in two dimensions.
 However, this specific inequality is only used in the three dimensional case to control \eqref{qq2Neu}. It can be verified (see euqations (92)-(97) in \cite{jeblick}) that it is not
necessary to control \eqref{qq2Neu} in two dimensions.
\end{remark}

%last item not needed
%\item
%$W_{\beta_1} \in \mathcal{V}_{\beta_1}$ follows directly from $R_{\beta_1} \leq C N^{-\beta_1}$.
%Furthermore, $ 0 \leq W_{\beta_1} (x) f_{\beta_1, 1} (x)  \leq W_{\beta_1}(x)$ implies
%$W_{\beta_1} f_{\beta_1, 1} \in \mathcal{V}_{\beta_1}$.
\end{enumerate}
\end{proof}

%%%%%%%%%%%%%%%%%%%%%%%%%%%%%%%%%%%%%%%%%%%%%%%%%%%%
%%%%%%%%%%%%%%%%%%%%%%%%%%%%%%%%%%%%%%%%%%%%%%%%%%%%
%%%%%%%%%%%%%%%%%%%%%%%%%%%%%%%%%%%%%%%%%%%%%%%%%%%%

\subsection{Nonnegativity of the Hamiltonian $H_U$}
\label{HposSection}
Next, we prove two important operator inequalities related to the Hamiltonian $H$, see
Corollary \ref{Hcor}. These inequalities will be used in order to show the inequalities
\eqref{condition2}, \eqref{condition3} and \eqref{condition4}.
%The next Lemma is stated for slightly more general potentials than those defined by assumption \ref{Vassumption}. This class of potentials is essentially the one considered in \cite{yin} (the constants $r_1,r_2$ and $R$ appearing below have the same meaning as in assumption \ref{Vassumption}).
\begin{lemma}
\label{HposLemma}
Let $U \in L_c^\infty( \mathbb{R}^3, \mathbb{R})$ fulfill assumption \ref{Vassumption} and
define
\begin{align*}
H_U=
-\sum_{k=1}^N
\Delta_k
+
\sum_{i<j=1}^N
U(x_i-x_j).
\end{align*}
%Let $U(x)=U^+(x)-U^-(x)$, where
%U^+,U^- \in L_c^\infty( \mathbb{R}^3, \mathbb{R})$ are spherically symmetric, such that $U^+(x),U^-(x) \geq 0$ (the supports of $U^+$ and $U^-$ %need not to be disjoint). 
%Assume that
%\begin{enumerate}
%\item
%For $R> r_2>0$, we have
%$\text{supp}(U^+)=B_{r_2}(0)$ and
%$\text{supp}(U^-)=B_{R}(0)$.
%\item
%There exists $\lambda^+>0$ and $r_1>0$, such that 
%$U^+(x) \geq \lambda^+$ for all $ x \in B_{r_1}(0)$.
%\item
%Let $\lambda^->0$ be defined such that
%$U^-(x) \leq \lambda^-$.
%\item
%We assume that there exists an $\epsilon>0$ such that
%\begin{align*}
%&
%\inf_{\phi \in C^1(\mathbb{R}^3, \mathbb{C}),  \phi(R)=1}
%\left(
%\int_{B_R(0)}
%\left(
%|\nabla_x \phi(x) |^2
%+
% n_1
%(2 U^+(x)-4(1+\epsilon) U^-(x))
%|\phi(x)|^2
%\right)d^3x
%\right)
%\geq 0,
%\\
%&
%\lambda^+ >  8 n_2 \lambda^- .
%\end{align*} 
%\end{enumerate}
Then 
\begin{align*}
H_U \geq 0 .
\end{align*}

\end{lemma}

In order to prove this Lemma, we first define
\begin{definition}
For $\tilde{R}\geq 2R$, where $R$ is defined as in assumption \ref{Vassumption}, let
for any $j,k=1, \dots, N$ with $j \neq k$
\be b_{j,k}:=\{(x_1,x_2,\ldots,x_N)\in
\mathbb{R}^{3N}: |x_j-x_k|\leq \tilde{R}\}\ee
$$
\overline{\mathcal{C}}_{l}:=\bigcup_{j,k\neq l}b_{j,k},\;\;\;\;\;\;\;\mathcal{C}_{l}:=\mathbb{R}^{3N}\backslash
\overline{\mathcal{C}}_{l}\;.
$$
\end{definition}

\begin{proof}
Let
\begin{align*}
H_{ \overline{\mathcal{C}}}=&
\sum_{k=1}^N
-\Delta _k  \mathds{1}_{\overline{\mathcal{C}}_k}
+
\sum_{i\neq j}  \mathds{1}_{\overline{\mathcal{C}}_j}
\frac{1}{2} U (x_i-x_j),
\\
H_{ \mathcal{C}}=&
\sum_{k=1}^N
-\Delta _k  \mathds{1}_{\mathcal{C}_k}
+
\sum_{i\neq j}  \mathds{1}_{\mathcal{C}_j}
\frac{1}{2} U (x_i-x_j)
.
\end{align*}

Note that
\begin{align*}
H_{ \mathcal{C}}
=
\sum_{k=1}^N
-\Delta _k  \mathds{1}_{\mathcal{C}_k}
+
\frac{1}{4}
\sum_{i\neq j} 
(
 \mathds{1}_{\mathcal{C}_j}
 +
  \mathds{1}_{\mathcal{C}_i}
)
\frac{1}{2} U (x_i-x_j)
\end{align*}
is a symmetric operator w.r.t. to exchange of coordinates $x_1, \dots, x_N$. Therefore, it suffices to
prove $
\laa \Psi, H_{ \mathcal{C}}
\Psi \raa \geq 0
$ for $\Psi \in L^2_s (\mathbb{R}^{3N}, \mathbb{C})$, since
\begin{align*}
\inf_{\Psi \in L^2 (\mathbb{R}^{3N}, \mathbb{C}), \|\Psi\|=1}
\laa \Psi, H_{ \mathcal{C}}
\Psi \raa
=
\inf_{\Psi \in L^2_s (\mathbb{R}^{3N}, \mathbb{C}), \|\Psi\|=1}
\laa \Psi, H_{ \mathcal{C}}
\Psi \raa .
\end{align*}
In order to prove $H_{ \mathcal{C}} \geq 0$, we show
$
K_1=
-\Delta _1  \mathds{1}_{\mathcal{C}_1}
+
\frac{1}{2}
\sum_{j=2}^N  \mathds{1}_{\mathcal{C}_1}
\frac{1}{2} U (x_1-x_j) \geq 0$ on $L^2_s(\mathbb{R}^{3N}, \mathbb{C})$.
Since
\begin{align*}
\inf_{\Psi \in L^2_s (\mathbb{R}^3, \mathbb{C}), \|\Psi\|=1}
\laa \Psi, H_{ \mathcal{C}}
\Psi \raa 
=&
\inf_{\Psi \in L^2_s (\mathbb{R}^3, \mathbb{C}), \|\Psi\|=1}
\sum_{i=1}^N
\laa \Psi,K_i
\Psi \raa 
\\
=&
N
\inf_{\Psi \in L^2_s (\mathbb{R}^3, \mathbb{C}), \|\Psi\|=1}
\laa \Psi,K_1
\Psi \raa 
\end{align*}
holds, it then follows $ H_{\mathcal{C}} \geq 0$.

The next Lemmata prove that
$
K_1\geq 0 \text{ and } H_{ \overline{\mathcal{C}}} \geq 0.
$
Since $H_U= \sum_{i=1}^N K_i+ H_{ \overline{\mathcal{C}}}$, it then follows $H_U \geq 0$.
\end{proof}

\begin{remark}
The reason to split the Hamiltonian as done above is the following:
The interaction $\mathds{1}_{\overline{\mathcal{C}}_j}
\frac{1}{2} U (x_i-x_j) $ is only nonzero, if, for fixed configurations $(x_1, \dots, x_N)$,  $x_i$ is closer than $R$ to 
$x_j$, but no other particles are closer than $R$ to neither $x_i$ nor $x_j$. Therefore, the set $\overline{\mathcal{C}}$ excludes those configurations, where three-particle interactions occur. The strategy to separate the configurations of possible three-particle interactions is well known within the literature, see e.g. \cite{lssy, yin} and references therein.
\end{remark}
Let us  restate an important Lemma.
\begin{lemma}\label{positiv V-w} 
$ $\\
\begin{itemize}
\item[(a)]
Let $R_{\beta_1}$ and $W_{\beta_1}$ be defined as in Lemma \ref{microscopic}.
Let $V$ fulfill assumption \ref{Vassumption}.
Then, for  any $\Psi\in H^1(\mathbb{R}^{3N}, \mathbb{C})$
$$\|\mathds{1}_{|x_1-x_2|\leq  R_{\beta_1}}\nabla_1\Psi\|^2+\frac{1}{2}\laa\Psi,
(V_{1}-W_{\beta_1})(x_1-x_2)\Psi\raa\geq0\;.$$
\item[(b)] 
Let $W_{\beta_1}$ be defined as in Lemma \ref{microscopic}. 
Let $V$ fulfill assumption \ref{Vassumption} and
let  $\Psi\in
L_s^2(\mathbb{R}^{3N}, \mathbb{C}) \cap H^1(\mathbb{R}^{3N}, \mathbb{C})$.
Then, for sufficiently large $N$
\begin{align*}
\|\mathds{1}_{\mathcal{B}_{1}}\mathds{1}_{\overline{\mathcal{A}}_{1}}\nabla_1\Psi \|^2
+
\frac{1}{2}
\laa\Psi ,\sum_{j\neq
1}\mathds{1}_{\mathcal{B}_{1}}\left(V_1-W_{\beta_1}\right)(x_1-x_j)\Psi \raa
\geq 0
\;.
\end{align*}
\end{itemize}

\end{lemma}
For nonnegative $V$,
the proof has been given in \cite{picklgp3d} for the three-dimensional case (see Lemma 5.1. (3)) and in \cite{jeblick} for the two-dimensional analog (see Lemma 7.10). The proof given in these works
is not using the nonnegativity of $V$ directly, but
is based on the fact that $f_{\beta_1,1}$ is a nonnegative function. Therefore, the proof is also applicable in our setting, using Lemma \ref{defAlemma}.

\begin{lemma}
\label{H_cPos}
Let $K_1$ and $H_{ \overline{\mathcal{C}}}$ be defined as above. Under the assumptions of
Lemma \ref{HposLemma}, we have
\begin{enumerate}
\item
\begin{align*}
K_1 \geq 0 \text{ on } L^2_s (\mathbb{R}^{3N}, \mathbb{C}).
\end{align*}
\item
\begin{align*}
H_{ \overline{\mathcal{C}}} \geq 0 \text{ on } L^2 (\mathbb{R}^{3N}, \mathbb{C}).
\end{align*}
\end{enumerate}
\end{lemma}

\begin{proof}
\begin{enumerate}
\item 
The proof of Lemma \ref{positiv V-w}, part (b) can be straightforwardly applied to
prove part (a),
see Lemma 5.1. (3) in \cite{picklgp3d} and Lemma 7.10 in \cite{jeblick}.
Note for the proof to be valid, it is important that $\mathds{1}_{\mathcal{C}_k}(x_1, \dots, x_N)$ excludes those configurations where the distance of two distinct particles $x_i$ and $x_j$, $i,j \neq k$ to $x_k$ is smaller than $R$, which is the radius of the support of $U$.
We refer the reader to \cite{jeblick, picklgp3d} for the details of the proof.
\item 
\begin{remark}
The proof of part (b) originates from Lemma 10. in \cite{yin}.
The author, however, does not introduce the set $\mathcal{C}_k$, but uses a slightly different technique to exclude three particle interactions.
For conceptual clarity, we adapt the proof of Lemma 10. in  \cite{yin} to our definition of $H_{ \overline{\mathcal{C}}}$. 
 Since the proof given by Jun Yin is very elegant in our opinion, parts of the following are taken verbatim from \cite{yin}.  
\end{remark}
Recall that
\[
H_{ \overline{\mathcal{C}}}=
\sum_{k=1}^N
-\Delta _k  \mathds{1}_{\overline{\mathcal{C}}_k}
+
\sum_{i\neq j}  \mathds{1}_{\overline{\mathcal{C}}_j}
\frac{1}{2} U (x_i-x_j).
\]
Assume first that $N$ is even, i.e., $N=2N_1$ with $N_1 \in \mathbb{N}$. 
Let
$P = (\pi_1,\pi_2)$ be a partition of ${1, ...,N}$ into two disjoint sets with $N_1$ integers in $\pi_1$ and $\pi_2$, respectively. 
Let 
 \beq\label{defv1122}
U_{1,1}=U_{2,2}=
U^+\geq 0,\,\,\,\,
U_{1,2}=2 U_1^+-4 U^-,
\eeq 
with $U_{1,2}^- = -4 U^-$, $U_{1,2}^+ = 2 U^+$.
It then follows 
\[
\frac14\big(U_{1,1,}+U_{2,1}+U_{2,2}\big)
=U
.\]
For each $P$, we define (for shorter notation, we will implicitly assume $i \neq j$ in the following)
\begin{eqnarray}
\nonumber
H_P= H_{(\pi_1,\pi_2)}\equiv&&
\sum_{j\in\pi_1} -2 \Delta_j \mathds{1}_{\overline{\mathcal{C}}_j}
+
\sum_{i,j\in\pi_1}
\mathds{1}_{\overline{\mathcal{C}}_j}
\frac{1}{2} U_{1,1}(x_i-x_j)
\\\nonumber
+&&
\sum_{i\in\pi_2, j\in\pi_1}
\mathds{1}_{\overline{\mathcal{C}}_j}
\frac{1}{2} U_{1,2}(x_i-x_j)+ 
\sum_{i,j\in\pi_2}\mathds{1}_{\overline{\mathcal{C}}_j}
\frac{1}{2} U_{2,2}(x_i-x_j).
\end{eqnarray}
Consequently,
  $U_{\alpha,\beta}$ denotes the interaction potential between particles in $\pi_\alpha$ and $\pi_\beta$.
Note that 
\begin{align*}
-&\sum_{P}\sum_{j\in\pi_1} \Delta_j\mathds{1}_{\overline{\mathcal{C}}_j}
=
-\sum_{j=1}^N \Delta_j \mathds{1}_{\overline{\mathcal{C}}_j}
\frac{1}{2}\sum_{P} ,
\\
&
\sum_{P}
\sum_{i,j\in\pi_1}
\mathds{1}_{\overline{\mathcal{C}}_j}
U_{1,1}(x_i-x_j)
=
\sum_{P}
\sum_{i,j\in\pi_2}
\mathds{1}_{\overline{\mathcal{C}}_j}
U_{2,2}(x_i-x_j) 
\\
=&
\sum_{i\neq j=1}^N
\mathds{1}_{\overline{\mathcal{C}}_j}
 U^+(x_i-x_j)
\frac{1}{4}\sum_{P} ,
\\
&
\sum_{P}
\sum_{i\in\pi_1, j \in \pi_2}
\mathds{1}_{\overline{\mathcal{C}}_j}
 U_{1,2}(x_i-x_j)
\\
=&
\sum_{i\neq j=1}^N
\mathds{1}_{\overline{\mathcal{C}}_j}
(
2 U^+(x_i-x_j)-4 U^-(x_i-x_j)
)
\frac{1}{4}\sum_{P} .
\end{align*}
Therefore,
\beq
H_{ \overline{\mathcal{C}}}=\sum_PH_P/\sum_P 1.
\eeq
Hence, for $N$ even, to obtain $H_{ \overline{\mathcal{C}}}\geq 0$, it is sufficient to prove that for  $\forall P$, $H_P\geq 0$. 
\medskip \\
If $N$ is odd, we divide $P=(\pi_2, \pi_2)$, with $N_1=(N-1)/2$ integers in $\pi_1$ and
$(N+1)/2$ integers in $\pi_2$. 
%Combinatorical argument
\\
Let $A_j$ be a one-particle operator and define, 
for any partition $P=(\pi_1, \pi_2)$, 
$\delta_{j \in \pi_1}$ such that $\delta_{j \in \pi_1}=1$ if $j \in \pi_1$, otherwise $0$.
Then
$
\sum_P \sum_{j \in \pi_1} A_j
=
\sum_{j=1}^N A_j \sum_P \delta_{j \in \pi_1} 
$.
Note that
\begin{align*}
\sum_P \delta_{j \in \pi_1} 
=
\frac{\sum_P \delta_{j \in \pi_1} }{\sum_P  } \sum_P 
=
\frac{
\binom{N-1}{\frac{N-3}{2}}
}
{
\binom{N}{\frac{N-1}{2}}
}
\sum_P 
=
\frac{1-\frac{1}{N}}{2}
\sum_P.
\end{align*}
Furthermore, for any two-particle operator $A_{i,j}$, we obtain, for $a,b \in \{1,2 \}$, 
\begin{align*}
\sum_P \sum_{i \in \pi_a, j \in \pi_b, i \neq j} A_{i,j}
=
\sum_{i \neq j =1}^N A_{i,j} \sum_P \delta_{i \in \pi_a} \delta_{j \in \pi_b} .
\end{align*}
Let $i \neq j$. With
\begin{align*}
&
\frac{1}{\sum_p}
\sum_P \delta_{i\in \pi_1} \delta_{j \in \pi_1} 
=
\frac{
\binom{N-2}{\frac{N-5}{2}}
}
{
\binom{N}{\frac{N-1}{2}}
}
=
\frac{1}{4} \left(1- \frac{3}{N} \right),
\;
&
\frac{1}{\sum_p}
\sum_P \delta_{i \in \pi_1} \delta_{j \in \pi_2}
=
\frac{
\binom{N-2}{\frac{N-3}{2}} 
}
{
\binom{N}{\frac{N-1}{2}}
}
=
\frac{1}{4}
\left(
1+ \frac{1}{N}
\right),
\\
&
\frac{1}{\sum_p}
\sum_P \delta_{i \in \pi_2} \delta_{j \in \pi_1}
=
\frac{
\binom{N-2}{\frac{N-3}{2}} 
}
{
\binom{N}{\frac{N-1}{2}}
}
=
\frac{1}{4}
\left(
1+ \frac{1}{N}
\right),
\;
&
\frac{1}{\sum_p}
\sum_P \delta_{i \in \pi_2} \delta_{j \in \pi_2} 
=
\frac{
\binom{N-2}{\frac{N-1}{2}}
}
{
\binom{N}{\frac{N-1}{2}}
}
=
\frac{1}{4}
\left(
1+ \frac{1}{N}
\right),
\end{align*}
%End combinatorical argument
it follows that 
\begin{align*}
-&\sum_{P}\sum_{j\in\pi_1} \Delta_j\mathds{1}_{\overline{\mathcal{C}}_j}
=
-
\frac{1-\frac{1}{N}}{2}
\sum_{j=1}^N \Delta_j \mathds{1}_{\overline{\mathcal{C}}_j}
\sum_{P} ,
\\
&
\sum_{P}
\sum_{i,j\in\pi_1}
\mathds{1}_{\overline{\mathcal{C}}_j}
U_{1,1}(x_i-x_j)
=
\frac{1}{4}\left(1- \frac{3}{N} \right)
\sum_{i\neq j=1}^N
\mathds{1}_{\overline{\mathcal{C}}_j}
 U^+(x_i-x_j)
 \sum_{P} ,
 \\
&
\sum_{P}
\sum_{i,j\in\pi_2}
\mathds{1}_{\overline{\mathcal{C}}_j}
U_{2,2}(x_i-x_j)
=
\frac{1}{4}\left(1+ \frac{1}{N} \right)
\sum_{i\neq j=1}^N
\mathds{1}_{\overline{\mathcal{C}}_j}
 U^+(x_i-x_j)
 \sum_{P} ,
\\
&
\sum_{P}
\sum_{i\in\pi_1, j \in \pi_2}
\mathds{1}_{\overline{\mathcal{C}}_j}
 U_{1,2}(x_i-x_j)
=
\frac{1}{4}\left(1+ \frac{1}{N} \right)
\sum_{i\neq j=1}^N
\mathds{1}_{\overline{\mathcal{C}}_j}
 U_{1,2}(x_i-x_j)
\sum_{P} .
\end{align*}
\medskip \\
For $N$ odd and $N$ large enough, the bound of  $H_P\geq 0$, $\forall P$ 
then implies, together with the assumption \ref{Vassumption} on $U$, that
$H_{ \overline{\mathcal{C}}} \geq 0$. 
\medskip\\
We will now prove  $H_P\geq 0$, $\forall P$.
The advantage to consider $H_P$ instead of $H_{ \overline{\mathcal{C}}}$ is that we can analyze
$H_P\geq 0$ for fixed configurations of $x_{i}$'s with $i\in\pi_2$. This pointwise estimate is sufficient, since there is no kinetic energy 
of the $\pi_2$-particles.
Since permutation of the labels in $\pi_1$ and $\pi_2$ is irrelevant, we can further assume that $\pi_1=\{1,\cdots,N_1\}$, $\pi_2=\{N_1+1,\cdots,N\}$. 
\medskip \\
Following the idea of \cite{yin}, for any fixed configuration $(x_{N_1+1}, \dots, x_N)$, we consider two cases:
\begin{itemize}
\item
If there are more than $m_1$ $\pi_2$-particles in a sphere of radius $R$ with $m_1 \geq 2 n_1$ , the positive interaction $U_{2,2}$, together with $U_{1,1}$ cancels the negative part of $U_{1,2}$. 
Recall that
$n_1$ is the number of cubes of side length $r_1/\sqrt{3}$ which are needed to cover a sphere of radius $R$. Therefore, if  $m_2$ $\pi_2$-particles are located in such a sphere, it is possible to derive that at least $\mathcal{O}(m_2^2/n_1)$ $\pi_2$-particles are closer than $r_1$
to each other. Therefore, if $m_1$ $\pi_1$-particles and $m_2$ $\pi_2$-particles are close to each other, the potential energy is of order $\mathcal{O}(m_1^2)+
\mathcal{O}(m_2^2)-\mathcal{O}(m_1m_2)$. This energy is positive, if the negative part of $U$ is small enough.
\item
If there are less than $2 n_1$ $\pi_2$-particles in a sphere of radius $R$, it is possible to
use assumption \ref{Vassumption}, \eqref{minenergyscatt}, that is
\begin{align*}
-
  \mathds{1}_{|x|\leq R}  
\Delta_x+
n_1
(2 U^+(x)-4 U^-(x)) 
\geq 
0.
\end{align*}
\end{itemize}
As in Definition \ref{cubedef}, we  
divide $\mathbb{R}^ 3$ into cubes  $C_n$ ($n\in\mathbb{N}$) of side length
  $ \frac{1}{\sqrt{3}} r_1$, such that the distance between to points $x_i, x_j \in C_n$ is not greater than $r_1$. 
Therefore, for $x_i, x_j \in C_n$ we have by assumption $U(x_i-x_j  ) \geq \lambda^+$.
Next, for fixed $x_{i}$, $i\in\pi_2$, for any $x \in \mathbb{R}^ 3$,  we define $G(x)$ as the set of $i$'s which satisfy $i\in\pi_2$ and $|x_i-x|\leq R$, i.e., 
\begin{eqnarray}
\label{defGx}
G(x)\equiv\{i\in \pi_2: |x_i-x|\leq R\}.
\end{eqnarray}
We denote $|G(x)|$ as the number of the elements of $G(x)$. Note that for $i, j \in G(x)$, it follows that $|x_i-x_j| \leq 2R$.
\par
 We denote $d(x, C_n)$ as the distance between the cube $C_n\subset \mathbb{R}^3$ and $x\in \mathbb{R}^3$. Since $|G(y)|$ is uniformly bounded ($|G(y)|\leq N_1$), 
there must exist a point $X(C_n)\in\mathbb{R}^3$ satisfying $d(X(C_n), C_n)\leq 2R$ and 
\begin{eqnarray}
\label{defGb}
|G(X(C_n))|=\max\{|G(y)|: d(y, C_n)\leq 2R \}.
\end{eqnarray}
We define $G(C_n)\equiv G(X(C_n))$. 
Let $ \mathds{1}_{ C_n} (x_j)$ denote the projection onto $C_n$ in the coordinate $x_j$.
Furthermore,
let $\Theta$ denote the usual Heaviside step function.
We prove
\begin{align*}
\mathcal{H}_1=&
\sum_{i,j \in \pi_2}
\mathds{1}_{\overline{\mathcal{C}}_j}
 U_{2,2}(x_i-x_j)
 +
 \sum_{i,j \in \pi_1}
\mathds{1}_{\overline{\mathcal{C}}_j}
 U_{1,1}(x_i-x_j)
 \\
 -&
 \sum_{n \in \mathbb{N}}
\Theta(|G(C_n)|-2n_1)
 \sum_{j \in \pi_1, i \in \pi_2}
 \mathds{1}_{ C_n} (x_j)
\mathds{1}_{\overline{\mathcal{C}}_j}
 U_{1,2}^-(x_i-x_j) \geq 0
 \\
 \mathcal{H}_{2,j} =&
   -2  \Delta_j \mathds{1}_{\overline{\mathcal{C}}_j}
+
 \sum_{ i \in \pi_2}
\mathds{1}_{\overline{\mathcal{C}}_j}
\frac{1}{2}
 U^+_{1,2}(x_i-x_j)
 \\
 -&
  \sum_{n \in \mathbb{N}}
\Theta(2n_1-|G(C_n)|)
 \sum_{ i \in \pi_2}
 \mathds{1}_{ C_n} (x_j)
\mathds{1}_{\overline{\mathcal{C}}_j}
\frac{1}{2}
 U_{1,2}^-(x_i-x_j)
  \geq 0.
\end{align*}
Note that this implies $H_p \geq 0$, since  $H_p= \frac{1}{2}\mathcal{H}_1+ \sum_{j\in\pi_1} \mathcal{H}_{2,j} $.

\par 
Proof of $\mathcal{H}_1 \geq 0$:
\\
First, we derive the lower bound on the total energy of $U_{2,2}$. With the definition of $G(C_n)=G(X(C_n))$, we know that the set $\{x_k: k\in G(C_n)\}$ can be covered by a sphere of radius $R$. So the number of the cubes which one need to cover this set is less than $n_1$. We denote these cubes as $C_{n_1}\cdots C_{n_m}$ $(m\leq n_1)$ and  assume the number of $i$'s satisfying $i\in G(C_n)$ and $x_i\in C_{n_k}$ is $a_{n_k}$. Because the side length of $C_{n_k}$ is equal to  $r_1/\sqrt{3}$, the distance between the two particles in the same cube is no more  than $r_1 $. Hence, we obtain, for $i \neq j$, 
\begin{align*}
\sum_{i,j\in G(C_n)}\theta_{r_1}(x_i-x_j)
\geq
\sum_{k=1}^m \sum_{
 i,j \in C_{n_k}
 }
=\sum_{k=1}^m \left[(a_{n_k})^2-(a_{n_k})\right]
\qquad \text{ and } \sum_{k=1}^m a_{n_k} = | G(C_n)|.
\end{align*}
Using Jensen's inequality, together with $ m \leq n_1$, 
\begin{align*}
\sum_{i,j\in G(C_n)}\theta_{r_1}(x_i-x_j)
&\geq  \frac{1}{2n_1} |G(C_n)|^2.
\end{align*}
 Note that
for fixed $i \in \pi_2$, the number of cubes $C_n$, which satisfy $i \in G(C_n)$ is less
than $n_2$. Since $U_{2,2}$ is nonnegative, we then obtain 
\begin{align*}
& \sum_{i,j \in \pi_2}
\mathds{1}_{\overline{\mathcal{C}}_j}
 U_{2,2}(x_i-x_j)
=
\sum_{n \in \mathbb{N}}
\sum_{i,j \in \pi_2}
\mathds{1}_{C_n} (x_i)
\mathds{1}_{\overline{\mathcal{C}}_j}
 U_{2,2}(x_i-x_j)
 \\
 \geq &
 \frac{1}{n_2}
 \sum_{n \in \mathbb{N}}
\sum_{i,j \in \pi_2, i \in G(C_n)}
\mathds{1}_{\overline{\mathcal{C}}_j}
 U_{2,2}(x_i-x_j)
 \\
  \geq  &
 \frac{1}{n_2}
 \sum_{n \in \mathbb{N}}
 \Theta( |G(C_n)|-2n_1)
\sum_{i,j \in G(C_n)}
\mathds{1}_{\overline{\mathcal{C}}_j}
 U_{2,2}(x_i-x_j)  \;.
  \end{align*}
Since $r_1<R$, it also follows that
$n_1\geq 2$. We then obtain
$\mathds{1}_{\overline{\mathcal{C}}_j}
 U_{2,2}(x_i-x_j)=
 U_{2,2}(x_i-x_j)$, whenever $i,j \in G(C_n)$ with $|G(C_n)| \geq 2n_1$.
Using $U_{2,2}(x)\geq \lambda^+ \Theta_{r_1}(x_i-x_j)$, we have with the estimates above
  \begin{align*}
\sum_{i,j \in \pi_2}
\mathds{1}_{\overline{\mathcal{C}}_j}
U_{2,2}(x_i-x_j)
\geq 
 \sum_{n \in \mathbb{N}}
 \Theta( |G(C_n)|-2n_1)
 \frac{\lambda^+}{2 n_1 n_2}|G(C_n)|^2 \;.
\end{align*}

Next, we derive the lower bound on  the interaction potential between particles in $\pi_1$. 
Let $\Pi_1(C_n)$ be defined as the set of $i$'s such that  $i\in\pi_1$ and $x_i\in C_n$. Let $|\Pi_1(C_n)|$ denote the number of the elements of $\Pi_1(C_n)$. If $x_i\in C_n$ and $|G(C_n)|\geq 1$, there must be a $k\in\pi_2$ satisfying $|x_i-x_k| \leq 2R$. Thus, for any $C_n$ we have that 
\begin{align*}
&
\sum_{i, j \in \pi_1}
\mathds{1}_{\overline{\mathcal{C}}_j}
U_{1,1}(x_i-x_j)
=
\sum_{n \in \mathbb{N}}
\sum_{i, j \in \pi_1}
\mathds{1}_{C_n} (x_i)
\mathds{1}_{\overline{\mathcal{C}}_j}
U_{1,1}(x_i-x_j)
\\
&\geq
\sum_{n \in \mathbb{N}}
\Theta (|G(C_n)|- 2n_1)
\sum_{i, j \in \Pi_1 (C_n)}
U_{1,1}(x_i-x_j) .
\end{align*}
For $i, j \in \Pi_1(C_n),\; i \neq j$, the distance between $x_i$ and $x_j$ is not more than $r_1$.
Hence,
\begin{align}
\label{50}
\sum_{i,j\in\Pi_1(C_n)}
U_{1,1}(x_i-x_j)\geq 
 \lambda^+
\bigg(|\Pi_1(C_n)|^2-|\Pi_1(C_n)|\bigg).
\end{align}
At last, we derive the lower bound on $U^-_{1,2}$.

 By the definitions of $|G(C_n)|$ and $U_{1,2}$, we have that  $\forall x\in C_n$,
 \[
 -
 \sum_{i\in\pi_2} U_{1,2}^-(x- x_i)\geq
 - 4 \lambda^-
 |G(C_n)|.\]
This yields to
\beq\label{51}
-
\sum_{j\in\Pi_1(C_n),\,\, i\in\pi_2}
\mathds{1}_{\overline{\mathcal{C}}_j}U_{2,1}^-(x_i-x_j)
\geq -4 \lambda^- |\Pi_1(C_n)| |G(C_n)|.
\eeq

We now consider
\begin{align}
&
\label{abschaetzung}
\sum_{i,j\in\Pi_1(C_n)}
U_{1,1}(x_i-x_j)
-
\sum_{j\in\Pi_1(C_n),\,\, i\in\pi_2}
\mathds{1}_{\overline{\mathcal{C}}_j}U_{1,2}^-(x_i-x_j)
\\
\nonumber
\geq &
\lambda^+ 
\bigg(|\Pi_1(C_n)|^2-|\Pi_1(C_n)|\bigg)
-4 \lambda^- |\Pi_1(C_n)| |G(C_n)|.
\end{align}
Using $ \lambda^- \leq \frac{1}{8 n_2}\lambda^+ $, we then obtain
for $|G(C_n)|\geq n_1$
\begin{align*}
\eqref{abschaetzung}
\geq
\lambda^+
\left(
|\Pi_1(C_n)|^2-|\Pi_1(C_n)|
	-\frac{1}{2n_2}|\Pi_1(C_n)| |G(C_n)|
\right).
\end{align*}
If $|\Pi_1(C_n)|=1$, we obtain for $|G(C_n)| \geq 2n_1 $ 
\begin{align*}
\eqref{abschaetzung} \geq
-
\lambda^+ 
\frac{ |G(C_n)|^2}{4 n_1 n_2}
 .
\end{align*}
For $|\Pi_1(C_n)|\geq 2$, we have
$
|\Pi_1(C_n)|^2-|\Pi_1(C_n)| \geq \frac{1}{2}|\Pi_1(C_n)|^2
$ and therefore, for $|G(C_n)| \geq 2n_1 $ 
\begin{align*}
\eqref{abschaetzung}
\geq
\frac{\lambda^+}{2}
\left(
|\Pi_1(C_n)|^2- 2|\Pi_1(C_n)|
\frac{1}{2n_2} |G(C_n)|
\right)
\geq
-
\frac{\lambda^+}{2}
\frac{1}{4(n_2)^2} |G(C_n)|^2.
\end{align*}
Since $n_2 \geq n_1$ holds, we then obtain for $|G(C_n)| \geq 2n_1 $ and for all
$|\Pi_1(C_n)| \in \mathbb{N}$ 
\begin{align*}
\eqref{abschaetzung} \geq
-
\lambda^+ 
\frac{ |G(C_n)|^2}{4 n_1 n_2}
 .
\end{align*}
Therefore, we obtain
\begin{align*}
\mathcal{H}_1
\geq
 \sum_{n \in \mathbb{N}}
  \Theta( |G(C_n)|-2n_1)
 \left(
 \frac{\lambda^+}{2 n_1 n_2}|G(C_n)|^2
-
\frac{ \lambda^+}{4 n_1 n_2} |G(C_n)|^2
 \right)
\geq 0 
 .
\end{align*}
\medskip 
Proof of $\mathcal{H}_{2,j} \geq 0$:
\\
Since there is no kinetic energy for the $\pi_2$ particles, we prove
$
\mathcal{H}_{2,j} \geq 0
$ for fixed $x_i$, $i \in \pi_2$.
Define
\begin{align}
\label{h2}
\tilde{\mathcal{H}}_{2, j}=
   -2  \Delta_j
+
 \sum_{ i \in \pi_2}
\frac{1}{2}
 U^+_{1,2}(x_i-x_j)
 -
  \sum_{n \in \mathbb{N}}
\Theta(2n_1-|G(C_n)|)
 \sum_{ i \in \pi_2}
 \mathds{1}_{ C_n} (x_j)
\frac{1}{2}
 U_{1,2}^-(x_i-x_j)
\end{align}
Note that
\begin{align*}
\mathcal{H}_{2, j}
=
\mathds{1}_{\overline{\mathcal{C}}_j}
\tilde{\mathcal{H}}_{2, j}
\end{align*}
and $ \mathds{1}_{\overline{\mathcal{C}}_j}$ commutes with $-\Delta_j$.
Hence, it suffices to prove $\tilde{\mathcal{H}}_{2, j} \geq 0$.
Let
\begin{align*}
\pi_2'
=
\lbrace
i \in \pi_2:
\exists C_n,
D(x_i,C_n) \leq R, |G(C_n)| \leq 2 n_1
\rbrace .
\end{align*}
For fixed $x_i$ an d $x_j$, if
$$
\Theta(2n_1-|G(C_n)|)
 \mathds{1}_{ C_n} (x_j)
\frac{1}{2}
 U_{1,2}^-(x_i-x_j) \neq 0 ,
 $$  it then follows $i \in \pi_2'$.
 Therefore,
 \begin{align*}
   \sum_{n \in \mathbb{N}}
\Theta(2n_1-|G(C_n)|)
 \sum_{ i \in \pi_2}
\frac{1}{2}
 \mathds{1}_{ C_n} (x_j)
 U_{1,2}^-(x_i-x_j)
 \leq
 \sum_{ i \in \pi'_2}
\frac{1}{2}
 U_{1,2}^-(x_i-x_j).
 \end{align*}
Since  $\pi_2' \subset \pi_2$, it follows that
\begin{align*}
\eqref{h2}
\geq
   -2  \Delta_j
+
 \sum_{ i \in \pi_2'}
\frac{1}{2}
\left(
 U^+_{1,2}(x_i-x_j)
 -
 U_{1,2}^-(x_i-x_j)
 \right).
\end{align*}
By the definition of $\pi_2'$, it follows that for any $x \in \mathbb{R}^3$
\begin{align*}
\sum_{i \in \pi_2'}
  \mathds{1}_{|x_i-x|\leq R}   \leq 2n_1.
\end{align*}
Under the assumptions on $U$, we obtain
\begin{align*}
\eqref{h2}
\geq
\frac{1}{n_1}
 \sum_{ i \in \pi_2'}
 \left(
    -
  \mathds{1}_{|x_i-x_j|\leq R}  
      \Delta_j
 + 
\frac{n_1}{2}
 U_{1,2}(x_i-x_j)
 \right)
 \geq
 0.
\end{align*}
\end{enumerate}

\end{proof}

\begin{corollary}
\label{Hcor}
Let $V$ fulfill assumption \ref{Vassumption}. Then, there exists $0<\epsilon<1$ such that
\begin{align}
\label{inequ1}
-\sum_{k=1}^N
\Delta_k
+
\sum_{i<j=1}^N
(V_1^+(x_i-x_j)
-
(1+ \epsilon)
V_1^-(x_i-x_j)
)
\geq 
0,
\\
\label{inequ2}
(1-\epsilon)
\sum_{k=1}^N
-\Delta _k  \mathds{1}_{\overline{\mathcal{B}}_k}
+
\sum_{i\neq j}  \mathds{1}_{\overline{\mathcal{B}}_j}
\frac{1}{2} V_1 (x_i-x_j) 
 \geq 0.
\end{align}
\end{corollary}
\begin{remark}
These operator inequalities are crucial in order to prove conditions
\eqref{condition2}, \eqref{condition3} and \eqref{condition4}, see below.
We do not except the persistence of condensation if \eqref{inequ1} and
\eqref{inequ2} were not true. In that case, one would rather expect the condensate to collapse in the limit $N \rightarrow \infty$ in finite time.
\end{remark}

\begin{proof}
By rescaling $Nx \rightarrow x$, the first inequality \eqref{inequ1} is equivalent to
$-\sum_{k=1}^N
\Delta_k
+
\sum_{i<j=1}^N
(V^+(x_i-x_j)
-
(1+ \epsilon)
V^-(x_i-x_j)
)
\geq 
0$.
Setting $U(x)= V^+(x)-(1+\epsilon) V^-(x)$, $U$ then fulfills the conditions of Lemma \ref{HposLemma} which implies the inequality above.
\medskip
\\
Setting 
$\overline{\mathcal{D}}_{j}:=\bigcup_{k,l\neq j}\{(x_1,x_2,\ldots,x_N)\in
\mathbb{R}^{3N}: |x_l-x_k|<NN^{-26/27}\}$, the second inequality is equivalent to
$$
(1-\epsilon)
\sum_{k=1}^N
-\Delta _k  \mathds{1}_{\overline{\mathcal{D}}_k}
+
\sum_{i\neq j}  \mathds{1}_{\overline{\mathcal{D}}_j}
\frac{1}{2} V (x_i-x_j) 
 \geq 0.
$$
Note that the set $\overline{\mathcal{D}}_{j}$ defined above fulfills $\tilde{R}=N^{1/27}>2R$.
Hence, 
Lemma \ref{H_cPos}, part (b) implies the second inequality \eqref{inequ2}, setting $U=\frac{1}{1-\epsilon}V$.

\end{proof}

\subsection{Proof of condition \eqref{condition2} and \eqref{condition3}}
\begin{lemma}\label{corest}
Let $V$ fulfill assumption \ref{Vassumption} and let $A_t \in L^\infty(\mathbb{R}^3, \mathbb{R})$. Then, for all $\Psi \in L_s^2 (\mathbb{R}^{3N}, \mathbb{C}) \cap H^2(\mathbb{R}^{3N},\mathbb{C})$
\begin{enumerate}
\item
\begin{align}
\|V_1(x_1-x_3)\Psi\|^ 2
\leq &
C \laa \Psi, H \Psi \raa+C N.
\end{align}
\item
\begin{align}
\|\nabla_1\Psi\|^ 2
\leq &
\frac{C}{N}
( \laa \Psi, H \Psi \raa+1).
\end{align}
\end{enumerate}
\end{lemma}

\begin{proof}
\begin{enumerate}
\item

Let, for $0<\epsilon<1$,
$$
H^{(\epsilon)}=
-\sum_{k=1}^N \Delta_k
+
\sum_{i<j}
(
V^+_1(x_i-x_j)- (1+\epsilon) V_1^-(x_i-x_j)
)
+
\sum_{k=1}^N A_t(x_k) .
$$
Since $V$ fulfills assumption \ref{Vassumption}, Corollary  \ref{Hcor} then implies
together with $A_t \in L^\infty(\mathbb{R}^3, \mathbb{R})$, $H^{(\epsilon)} \geq -CN$.
We then obtain
$$
\epsilon \sum_{i<j=1}^ N V_1^{-} (x_i-x_j)
\leq
H+CN.
$$
Furthermore
\begin{align*}
 \sum_{i<j=1}^ N V_1^{+} (x_i-x_j)
\leq
H
+
 \sum_{i<j=1}^ N V_1^{-} (x_i-x_j)
 +
 N \| A_t\|_\infty
 \leq
\left( 1+\frac{1}{\epsilon} \right) H
+
C N .
\end{align*}
Thus,
\begin{align*}
\|V_1(x_1-x_3)\Psi\|^ 2
\leq &
\|V_1\|_{\infty}
(
\laa \Psi, V_1^+(x_1-x_3) \Psi \raa
+
\laa \Psi, V_1^-(x_1-x_3) \Psi \raa
)
\\
\leq &
C
\left(
\laa \Psi,
 \sum_{i<j=1}^ N V^+_1 (x_i-x_j)
\Psi \raa
+
\laa \Psi,
 \sum_{i<j=1}^ N V^-_1 (x_i-x_j)
\Psi \raa
\right)
\\
\leq &
C \laa \Psi, H \Psi \raa+C N.
\end{align*}
\item
We use
\begin{align*}
-CN \leq H^{(\epsilon)}
\leq
(1+\epsilon)
\left(
\frac{-1}{1+\epsilon}
\sum_{k=1}^N \Delta_k
+
\sum_{i<j}
V_1(x_i-x_j)
+
\sum_{k=1}^N \frac{1}{1+\epsilon} A_t(x_k) 
\right).
\end{align*}
Let $\mu= 1- \frac{1}{1+\epsilon}>0$.  Using $A_t \in L^{\infty} (\mathbb{R}^ 3, \mathbb{R})$, we then obtain
\begin{align*}
-\mu 
\sum_{k=1}^N \Delta_k
\leq
H+ CN.
\end{align*}

\end{enumerate}

\end{proof}
Using Lemma \ref{corest} together with $ 
\frac{\laa \Psi_t, H \Psi_t \raa}{N} \leq C$, we then obtain condition \eqref{condition2} and \eqref{condition3}.

\subsection{Proof of condition \eqref{condition4}}
We will first restate a Lemma which we will need in the following.
\begin{proposition}\label{propo}
Let $\Omega \in H^1(\mathbb{R}^{3N}, \mathbb{C})$. Then,
for all $j\neq k$ 
$$\|\mathds{1}_{\overline{\mathcal{B}}_{j}}\Omega\|\leq  CN^{-7/54}\|\nabla_j\Omega\|\;.$$
\end{proposition}
\begin{proof}
The proof of this Lemma, which is a direct consequene of Sobolev's inequality, can be found in \cite{picklgp3d}, Proposition A.1. for the three dimensional case and 
in \cite{jeblick}, Lemma 7.4. for the two dimensional analog (note that the set 
$\overline{\mathcal{B}}_{j}$ and the
respective $N$-dependent bound are different in two dimensions.).
\end{proof}

\begin{lemma}
\label{energylemma}
Assume $V$ fulfills assumption \ref{Vassumption}.
Then, for any $\Psi \in L^2_s(\mathbb{R}^{3N}, \mathbb{C}) \cap H^2(\mathbb{R}^{3N}, \mathbb{C})$ and any
$\phi\in H^2(\mathbb{R}^{3}, \mathbb{C})$  there exists a $\eta>0$  such that
\begin{enumerate}
\item
\begin{align*}\|\mathds{1}_{\mathcal{A}_{1}}\nabla_1q^{\phi}_1
\Psi \|^2 \leq  C \left(\laa\Psi,\widehat{n}^{\phi}\Psi\raa+N^{-\eta}\right)+\left|\mathcal{E}(\Psi)-\mathcal{E}^{GP}(\phi)\right|.
 \end{align*}
\item
 \begin{align*}\|\mathds{1}_{\overline{\mathcal{B}}_{1}}\nabla_1
\Psi \|^2 \leq  C \left(\laa\Psi,\widehat{n}^{\phi}\Psi\raa+N^{-\eta}\right)+\left|\mathcal{E}(\Psi)-\mathcal{E}^{GP}(\phi)\right|.
 \end{align*}
\end{enumerate}
\end{lemma}
\begin{remark}
For nonnegative potentials, the proof of Lemma \ref{energylemma} was given in
Lemma 5.2. in \cite{picklgp3d} for the three dimensional case and in Lemma 7.9 in \cite{jeblick} for the two dimensional case. For potentials which fulfill assumption \ref{Vassumption} we use Corollary \ref{Hcor} in order to obtain the same bound.
\end{remark}

\begin{proof}
Let us first split up the energy difference. Since $\Psi \in L^2_s(\mathbb{R}^{3N},\mathbb{C})$ is symmetric,
\begin{align*}
\mathcal{E}(\Psi)-\mathcal{E}^{GP}(\phi)&=
\|\nabla_1\Psi \|^2+
(N-1)
\laa \Psi, V_1(x_1-x_2)
\Psi \raa
\\&-\|\nabla\phi \|^2-2a\|\phi^2\|^2+
\laa \Psi, A_t \Psi \raa
-
\langle \phi, A_t \phi \raa.
\end{align*}
Let $W_{\beta_1}$ be defined as in Lemma \ref{microscopic}
for some $\beta_1$. Then,
\begin{align*}
\mathcal{E}(\Psi)-\mathcal{E}^{GP}(\phi)=&\|\mathds{1}_{\mathcal{A}_{1}}\nabla_1\Psi\|^2
 + \|\mathds{1}_{\overline{\mathcal{B}}_{1}}\mathds{1}_{\overline{
\mathcal{A}}_{1}}\nabla_1\Psi \|^2+\|\mathds{1}_{\mathcal{B}_{1}}\mathds{1}_{\overline{\mathcal{A}}_{1}}\nabla_1\Psi \|^2 \nonumber\\&
+(N-1)
\laa \Psi,
\mathds{1}_{\overline{\mathcal{B}}_{1}}V_1(x_1-x_2)\Psi \raa
\\&+\laa\Psi ,\sum_{j\neq
1}\mathds{1}_{\mathcal{B}_{1}}\left(V_1-W_{\beta_1}\right)(x_1-x_j)\Psi \raa
\\& +\laa\Psi ,\sum_{j\neq
1}\mathds{1}_{\mathcal{B}_{1}}W_{\beta_1}(x_1-x_j)\Psi \raa
-\|\nabla\phi \|^2-2a\|\phi^2\|^2
\\&+\laa\Psi A_t\Psi\raa-\langle\phi A_t\phi\rangle
\;.
\end{align*}
Using that $q_1=1-p_1$, we obtain for $0<\epsilon<1$,
\begin{align} &
\label{deltaE}
\mathcal{E}(\Psi)-\mathcal{E}^{GP}(\phi)
\\
=&
\epsilon
\left(
 \|\mathds{1}_{\mathcal{A}_{1}}\nabla_1q_1\Psi \|^2
  +
 \|\mathds{1}_{\overline{\mathcal{B}}_{1}}
   \mathds{1}_{\overline{\mathcal{A}}_{1}} \nabla_1 \Psi \|^2
\right)
\\
\label{uno}
+&2 \Re\left(\laa\nabla_1q_1\Psi ,
\mathds{1}_{\mathcal{A}_{1}}
\nabla_1p_1\Psi \raa\right)
\\
\label{dos}
+& \|\mathds{1}_{\mathcal{B}_{1}}\mathds{1}_{\overline{\mathcal{A}}_{1}}\nabla_1\Psi \|^2
+
\frac{1}{2}
\laa\Psi ,\sum_{j=2}^N\mathds{1}_{\mathcal{B}_{1}}\left(V_1-W_{\beta_1}\right)(x_1-x_j)\Psi \raa
\\
\label{tres}
+&
\frac{N-1}{2}\laa\Psi ,\mathds{1}_{\mathcal{B}_{1}}p_1p_2W_{\beta_1}(x_1-x_2)p_1p_2\mathds{1}_{\mathcal{B}_{1}}\Psi\raa -\frac{a}{2}\|\phi^2\|^2
\\
\label{cuatro}
+&(N-1)\Re\laa\Psi ,\mathds{1}_{\mathcal{B}_{1}}(1-p_1p_2)W_{\beta_1}(x_1-x_2)p_1p_2\mathds{1}_{\mathcal{B}_{1}}\Psi\raa
\\
+&\frac{N-1}{2}\laa\Psi ,\mathds{1}_{\mathcal{B}_{1}}(1-p_1p_2)W_{\beta_1}(x_1-x_2)(1-p_1p_2)\mathds{1}_{\mathcal{B}_{1}}\Psi\raa
\\
\label{cinco}
+& \|\mathds{1}_{\mathcal{A}_{1}}\nabla_1p_1\Psi \|^2-\|\nabla\phi \|^2
\\
\label{seis}
+&\laa\Psi, A_t(x_1)\Psi\raa-\langle\phi , A_t\phi\rangle
 \\
 \label{new1}
 +&
( 1-\epsilon)
\left(
 \|\mathds{1}_{\mathcal{A}_{1}}\nabla_1q_1\Psi \|^2
 +
 \|\mathds{1}_{\overline{\mathcal{B}}_{1}}
 \mathds{1}_{\overline{\mathcal{A}}_{1}}\nabla_1\Psi \|^2
\right)
\\
\label{new2}
+&
\frac{N-1}{2}
\laa \Psi, 
\mathds{1}_{\overline{\mathcal{B}}_{1}} V_1(x_1-x_2)\Psi \raa
\;.
\end{align}
It has been shown in \cite{picklgp3d} that for some
suitable chosen $0<\beta_1<1$
there exists an $\eta >0$ such that
\begin{align*}
|
\eqref{deltaE}
|
+
|
\eqref{uno}
|
+
|
\eqref{tres}
|
+
|
\eqref{cinco}
|
+
|
\eqref{seis}
|
\leq
C \left(\llaa\Psi,\widehat{n}^{\phi}\Psi\rraa+N^{-\eta}\right)+\left|\mathcal{E}(\Psi)-\mathcal{E}^{GP}(\phi)\right|
.
\end{align*}
Since $
\eqref{dos}
\geq 0, 
\eqref{cuatro} \geq 0
$,
we are left to control \eqref{new1} and \eqref{new2}
in order to show
$$
\epsilon
\left(
 \|\mathds{1}_{\mathcal{A}_{1}}\nabla_1q_1\Psi \|^2
  +
 \|\mathds{1}_{\overline{\mathcal{B}}_{1}}
   \mathds{1}_{\overline{\mathcal{A}}_{1}} \nabla_1 \Psi \|^2
\right)
\leq
C \left(\llaa\Psi,\widehat{n}^{\phi}\Psi\rraa+N^{-\eta}\right)+\left|\mathcal{E}(\Psi)-\mathcal{E}^{GP}(\phi)\right|.
$$ For nonnegative potentials,
the trivial bound $\eqref{new1} + \eqref{new2} \geq 0$ is  sufficient in order to prove Lemma
\ref{energylemma}. For potentials fulfilling assumption \ref{Vassumption}, we use
\begin{align*}
\eqref{new1}
+
\eqref{new2}
=&
( 1-\epsilon)
\left(
 \|\mathds{1}_{\mathcal{A}_{1}}\mathds{1}_{\overline{\mathcal{B}}_{1}}\nabla_1\Psi \|^2
 +
 \|\mathds{1}_{\overline{\mathcal{B}}_{1}}
 \mathds{1}_{\overline{\mathcal{A}}_{1}}\nabla_1\Psi \|^2
\right)
+
\frac{N-1}{2}
\laa \Psi, 
\mathds{1}_{\overline{\mathcal{B}}_{1}} V_1(x_1-x_2)\Psi \raa
\\
-&(1-\epsilon)
 2 \Re
 \left(
 \laa \nabla_1\Psi,
 \mathds{1}_{\mathcal{A}_{1}}\mathds{1}_{\overline{\mathcal{B}}_{1}}\nabla_1 p_1\Psi \raa
 \right)
 \\
 +&
(1- \epsilon) 
\left(
 \|\mathds{1}_{\mathcal{A}_{1}} \mathds{1}_{\mathcal{B}_{1}}\nabla_1q_1\Psi \|^2
 +
 \|\mathds{1}_{\mathcal{A}_{1}}\mathds{1}_{\overline{\mathcal{B}}_{1}}\nabla_1 p_1\Psi \|^2
 \right).
\end{align*}
We will estimate each line separately.
The third line is positive.
Using Proposition \ref{propo}, we obtain
\begin{align*}
\|\mathds{1}_{\mathcal{A}_{1}}\mathds{1}_{\overline{\mathcal{B}}_{1}}\nabla_1 p_1\Psi \|
\leq
\|\mathds{1}_{\overline{\mathcal{B}}_{1}}\nabla_1 p_1\Psi \|
\leq
C
N^{-7/54}
\| \Delta_1 p_1\Psi \|.
\end{align*}
This implies for the second line
\begin{align*}
 |2 \Re
 \left(
 \laa  \nabla_1 \Psi,
 \mathds{1}_{\overline{\mathcal{B}}_{1}} \mathds{1}_{\mathcal{A}_{1}}\nabla_1 p_1\Psi \raa
 \right)
 |
 \leq
C N^{-7/54}
\;.
\end{align*}
Focusing on the first term, we obtain with Corollary \ref{Hcor}
\begin{align*}
&( 1-\epsilon)
\left(
 \|\mathds{1}_{\mathcal{A}_{1}}\mathds{1}_{\overline{\mathcal{B}}_{1}}\nabla_1\Psi \|^2
 +
 \|\mathds{1}_{\overline{\mathcal{B}}_{1}}
 \mathds{1}_{\overline{\mathcal{A}}_{1}}\nabla_1\Psi \|^2
\right)
+
\frac{N-1}{2} \laa \Psi, 
\mathds{1}_{\overline{\mathcal{B}}_{1}} V_1(x_1-x_2)\Psi \raa
\\
=&
\frac{1}{N}
\laa
 \Psi,
 \left(
(1-\epsilon)
\sum_{k=1}^N
-\Delta _k  \mathds{1}_{\overline{\mathcal{B}}_k}
+
\sum_{i\neq j}  \mathds{1}_{\overline{\mathcal{B}}_j}
\frac{1}{2} V_1 (x_i-x_j) \Psi \raa
\right)
 \geq 0
 \;.
\end{align*}

We have therefore shown
\begin{align*}
 \|\mathds{1}_{\mathcal{A}_{1}}\nabla_1q_1\Psi \|^2
 +
 \|\mathds{1}_{\overline{\mathcal{B}}_{1}}\mathds{1}_{\overline{
\mathcal{A}}_{1}}\nabla_1\Psi \|^2
\leq
C \left(\llaa\Psi,\widehat{n}^{\phi}\Psi\rraa+N^{-\eta}+\left|\mathcal{E}(\Psi)-\mathcal{E}^{GP}(\phi)\right| \right).
\end{align*}
Note that
\begin{align*}
 \|\mathds{1}_{\overline{\mathcal{B}}_{1}}
 \nabla_1q_1\Psi \|^2
 =&
  \|
  \mathds{1}_{\overline{\mathcal{A}}_{1}}
  \mathds{1}_{\overline{\mathcal{B}}_{1}}
 \nabla_1q_1\Psi \|^2
 + 
 \|
\mathds{1}_{\mathcal{A}_{1}} 
 \mathds{1}_{\overline{\mathcal{B}}_{1}}
 \nabla_1q_1\Psi \|^2
 \\
 \leq &
   \|
  \mathds{1}_{\overline{\mathcal{A}}_{1}}
  \mathds{1}_{\overline{\mathcal{B}}_{1}}
 \nabla_1
( 1-p_1)
 \Psi \|^2
 + 
 \|
\mathds{1}_{\mathcal{A}_{1}} 
 \nabla_1q_1\Psi \|^2
 \\
 \leq &
2    \|
  \mathds{1}_{\overline{\mathcal{A}}_{1}}
  \mathds{1}_{\overline{\mathcal{B}}_{1}}
 \nabla_1
 \Psi \|^2
+
2
\|
  \mathds{1}_{\overline{\mathcal{A}}_{1}}
  \mathds{1}_{\overline{\mathcal{B}}_{1}}
 \nabla_1 p_1
 \Psi \|^2
 + 
  \|
\mathds{1}_{\mathcal{A}_{1}} 
 \nabla_1q_1\Psi \|^2.
\end{align*}
Using
$
 \|\mathds{1}_{\overline{\mathcal{B}}_{1}}
 \mathds{1}_{\overline{
\mathcal{A}}_{1} }\nabla_1p_1\Psi \|\leq 
\|\mathds{1}_{\overline{\mathcal{B}}_{1}}\nabla_1p_1\Psi \|
\leq
C N^{-7/54} \| \Delta_1 p_1 \Psi \|
$, we then obtain the Lemma.

\end{proof}

\section*{Acknowledgments}
We are grateful to Nikolai Leopold and Robert Seiringer for pointing out to us 
the results of \cite{yin}. We also thank Phillip Grass for helpful remarks.
M.J. gratefully acknowledges financial support by the German National Academic Foundation.

\end{document}